\documentclass[a4paper,UKenglish,cleveref, autoref, thm-restate]{lipics-v2021}

\usepackage{apptools}
\newcommand{\restateref}[1]{\IfAppendix{\hyperref[#1]{$\star$}}{\hyperref[#1*]{$\star$}}}

\crefformat{footnote}{#2\footnotemark[#1]#3}

\usepackage{xspace} 
\usepackage{enumitem}

\newcommand{\OLPlong}{\textsc{Ordered Level Planarity}\xspace}
\newcommand{\OLP}{\textsc{OLP}\xspace}
\newcommand{\PLPlong}{\textsc{Partial Level Planarity}\xspace}
\newcommand{\PLP}{\textsc{PLP}\xspace}
\newcommand{\CLPlong}{\textsc{Constrained Level Planarity}\xspace}
\newcommand{\CLP}{\textsc{CLP}\xspace}

\hideLIPIcs  

\title{Constrained Level Planarity is FPT with Respect to the Vertex Cover Number} 

\author{Boris Klemz}{Universität Würzburg, Germany }{firstname +++ dot +++ lastname +++ at +++ uni +++ minus +++ wuerzburg +++ dot +++ de}{https://orcid.org/0000-0002-4532-3765}{
}

\author{{Marie Diana} Sieper}{Universität Würzburg, Germany}{marie.sieper@uni-wuerzburg.de}{https://orcid.org/0009-0003-7491-2811}{}

\authorrunning{B.\ Klemz and M.\ D.\ Sieper} 

\Copyright{Boris Klemz and Marie Diana Sieper} 



\nolinenumbers 

\EventEditors{John Q. Open and Joan R. Access}
\EventNoEds{2}
\EventLongTitle{42nd Conference on Very Important Topics (CVIT 2016)}
\EventShortTitle{CVIT 2016}
\EventAcronym{CVIT}
\EventYear{2016}
\EventDate{December 24--27, 2016}
\EventLocation{Little Whinging, United Kingdom}
\EventLogo{}
\SeriesVolume{42}
\ArticleNo{23}

\begin{document}

\maketitle

\begin{abstract}
The problem \textsc{Level Planarity} asks for a crossing-free
drawing of a graph in the plane such that vertices are placed at prescribed y-coordinates (called \emph{levels}) and such that every edge is realized as a y-monotone curve.
In the variant \textsc{Constrained Level Planarity},
each level $y$ is equipped with a partial order $\prec_y$ on its vertices and in the desired drawing the left-to-right order of vertices on level~$y$ has to be a linear extension of $\prec_y$.
\textsc{Constrained Level Planarity} is known to be a remarkably difficult problem:
previous results by Klemz and Rote [ACM Trans.\ Alg.'19] and by Br\"uckner and Rutter [SODA'17] imply that it remains NP-hard even when restricted 
to graphs whose tree-depth and feedback vertex set number are bounded by a constant and even when the instances are additionally required to 
be either \emph{proper}, meaning that each edge spans two consecutive levels, or \emph{ordered}, meaning that all given partial orders are total orders.
In particular, these results rule out the existence of FPT-time (even XP-time) algorithms with respect to these and related graph parameters (unless P=NP).
However, the parameterized complexity of \textsc{Constrained Level Planarity} with respect to the vertex cover number of the input graph remained open.

In this paper, we show that \textsc{Constrained Level Planarity} can be solved in FPT-time when parameterized by the vertex cover number.
In view of the previous intractability statements, our result is best-possible in several regards:
a speed-up to polynomial time or a generalization to the aforementioned smaller graph parameters is not possible, even if restricting to proper or ordered instances.
\end{abstract}

\section{Introduction}
\label{sec:intro}

A large body of literature related to graph drawing is dedicated to so-called upward planar drawings, which provide a natural way of  visualizing a partial order on a set of items.
An \emph{upward planar drawing} of a directed graph is a crossing-free drawing in the plane where every edge~$e=(u,v)$ is realized as a y-monotone curve 
that goes upwards from~$u$ to~$v$, i.e., the y-coordinate strictly increases when traversing~$e$ from~$u$ towards~$v$.
The most classical
computational
problem
in this context is \textsc{Upward Planarity}:
given a directed graph, decide whether it admits an upward planar drawing.
The standard version of this problem is NP-hard~\cite{DBLP:journals/siamcomp/GargT01},
but, if the y-coordinate of each vertex is prescribed, it can be solved in polynomial time~\cite{DBLP:journals/tsmc/BattistaN88,DBLP:conf/gd/HeathP95,DBLP:conf/gd/JungerLM98},
which suggests that a large part of the challenge 
of \textsc{Upward Planarity}
comes from choosing an appropriate y-coordinate for each vertex.
However, when both the y-coordinate and the x-coordinate of each vertex are prescribed, the problem is yet again NP-hard~\cite{DBLP:journals/talg/KlemzR19},
indicating that another part of the challenge comes from drawing the edges in a y-monotone non-crossing fashion while respecting the given or chosen coordinates of their endpoints.
The paper at hand is concerned with the parameterized complexity of a generalization of the latter of these two variants of \textsc{Upward Planarity}, which is known as \textsc{Constrained Level Planarity}.
It is expressed in terms of so-called level graphs, which are defined next;
we adopt the notation and terminology used in~\cite{DBLP:journals/talg/KlemzR19}.

\subparagraph{Level Planarity.} A \emph{level graph}~$\mathcal G=(G,\gamma)$ is a directed graph~$G=(V,E)$ together with a \emph{level assignment}, 
which is a function\footnote{\label{fnote:1}
Traditionally, in the literature, the level assignment~$\gamma$ is defined as a surjective function that maps to an integer interval $\lbrace 1,2,\dots ,h\rbrace$;
it merely acts as a convenient way to encode a total preorder on $V$.
It is well known that the traditional and our (more general) definition are polynomial-time equivalent: 
algorithms designed assuming the classical definition can also be applied in the more general context: one simple has to first sort the vertices by y-coordinates and then apply the traditional algorithm using the sorting-ranks as y-coordinates.
We are using the given general definition as it eases the description of our algorithms; though, specific polynomial runtimes obtained in previous work that are stated in our introduction assume the classical definition.
} $\gamma\colon V\rightarrow \mathbb R$ where~$\gamma (u)<\gamma (v)$ for every edge~$(u,v)\in E$.
For every $i\in \mathbb R$ where $V_i=\lbrace v\in V\mid \gamma (v)=i\rbrace$ is non-empty, the set $V_i$ is called a (the~$i$-th) \emph{level} of~$\mathcal G$.
The \emph{width} of level $V_i$ is~$|V_i|$.
The \emph{level-width} of~$\mathcal G$ is the maximum width of any level in~$\mathcal G$ and the \emph{height} of~$\mathcal G$ is the number of (non-empty) levels.
A \emph{level planar drawing} of~$\mathcal G$ is an upward planar drawing of~$G$ where the y-coordinate of each vertex~$v$ is~$\gamma (v)$. 
We use~$L_i$ to denote the horizontal line with y-coordinate~$i$.
The level graph~$\mathcal G$ is called \emph{proper} if every edge spans two consecutive levels, that is, for every edge~$(u,v)\in E$ there is no level $V_j$ with $\gamma(u)<j<\gamma(v)$.
The problem \textsc{Level Planarity} asks whether a given level graph admits a level planar drawing.
In a series of papers~\cite{DBLP:journals/tsmc/BattistaN88,DBLP:conf/gd/HeathP95,DBLP:conf/gd/JungerLM97,DBLP:conf/gd/JungerLM98},
it was shown that \textsc{Level Planarity} can be solved in linear\cref{fnote:1} time;
we refer to~\cite{fulek2013hanani} for a more detailed discussion of the history of the corresponding algorithm and of alternative approaches to solve \textsc{Level Planarity}.

\subparagraph{Constrained Level Planarity.}
In 2017, Br\"uckner and Rutter~\cite{DBLP:conf/soda/BrucknerR17} and Klemz and Rote~\cite{DBLP:journals/talg/KlemzR19} independently 
introduced and studied two closely related variants of \textsc{Level Planarity},
which are defined in the following.
A \emph{constrained level graph} $\mathcal G=(G,\gamma,(\prec_i)_{i})$ is a triple 
corresponding to a level graph $(G, \gamma)$  equipped with a family $(\prec_i)_{i}$ containing, for each level~$V_i$, a partial order~$\prec_i$ on the 
vertices~$V_i$.
A \emph{constrained} level planar drawing of $\mathcal G$ is a level planar drawing of $(G, \gamma)$ where, for each level $V_i$, the left-to-right 
order of the vertices~$V_i$ corresponds to a linear extension of~$\prec_i$.
For a pair of vertices $u,v\in V_i$ with $u\prec_i v$, we refer to $u\prec_i v$ as a \emph{constraint} on $u$ and~$v$.
The problem \CLPlong (\CLP) asks whether a given constrained level graph admits a constrained level planar drawing.
\OLPlong (\OLP) corresponds to the special case of \CLP where the given partial orders are total orders, 
which is polynomial time equivalent to prescribing the x-coordinate (in addition to the y-coordinate) of each vertex.

Klemz and Rote~\cite{DBLP:journals/talg/KlemzR19} established a complexity dichotomy for \OLP with respect to both the maximum degree and the level-width.
In particular, they showed that \OLP is NP-hard even when restricted to the case where $\mathcal G$ has a level-width 
of~$2$ and the underlying undirected graph of~$G$ is a disjoint union of paths, i.e., a graph of maximum degree~$2$, path-width (and tree-width) 
$1$, and feedback vertex/edge set number $0$.
In fact, with a simple modification\footnote{
In the variable gadget of every variable~$u_j$, one can remove the subdivision vertices of the tunnels with index larger than~$j$.
This modification does not influence the realizability of the instance since the left-to-right order of all tunnels is already fixed due to the subdivision vertices on level~$\ell_0$.
} to their construction, the underlying undirected graph produced by the reduction becomes a disjoint union of paths with constant length, implying that even the tree-depth is bounded.
(The definitions of all these classical graph parameters can be found, e.g., in~\cite{DBLP:books/sp/CyganFKLMPPS15}.)
It follows that \CLP is NP-hard in the same scenario and when, additionally, each of the prescribed partial orders $\prec_i$ is a total order.
\OLP is (trivially) solvable in linear\cref{fnote:1} time when restricted to proper instances~\cite{DBLP:journals/talg/KlemzR19}.
In contrast, an instance of \CLP can always be turned into an equivalent proper instance by subdividing each 
edge on each level it passes through without introducing any constraints on the resulting subdivision vertices~\cite{DBLP:conf/soda/BrucknerR17}.
Hence, \CLP is NP-hard even in the proper case.
Independently, Br\"uckner and Rutter~\cite{DBLP:conf/soda/BrucknerR17} also presented a proof for the NP-hardness of \CLP, which relies on a very different strategy.
It is not obvious whether the graphs produced by their construction have bounded tree-width, 
however, it is not difficult to see\footnote{
In the strongly NP-hard \textsc{3-Partition} problem~\cite{Garey:1990}, one has to partition $3n$ positive integers 
$B/4<a_1,a_2,\dots,a_{3n}<B/2$ of total sum~$nB$ into $n$ triples (or \emph{buckets}) of sum (or \emph{size})~$B$.
To reduce to \CLP, one can simulate a bucket of size~$B$ as a sequence of~$B$ consecutive sockets and a 
number~$a_i$ as~$a_i$ plugs that are connected in a star-like fashion to a common ancestor~$v_i$ located above all these plugs.
Finally, all ancestors~$v_i$ and all sockets are connected in a star-like fashion to a common root vertex.
} that the socket/plug gadget used in their reduction can be utilized in the context of a reduction from \textsc{3-Partition} to show that 
\CLP remains NP-hard for proper instances whose underlying undirected graph is a single (rooted) tree of constant depth.
In fact, the unpublished full version of~\cite{DBLP:conf/soda/BrucknerR17} features such a construction~\cite{ignaz-pc}.

On the positive side, Br\"uckner and Rutter~\cite{DBLP:conf/soda/BrucknerR17} presented a polynomial time algorithm for the special case of 
\CLP where the input graph~$G$ has a single source.
They further improved the runtime of this algorithm in~\cite{DBLP:conf/isaac/BrucknerR20}.
Very recently, Bla\v{z}ej, Klemz, Klesen, Sieper, Wolff, and Zink studied the parameterized complexity of \CLP and \OLP with respect to the height of the input graph~\cite{fewLevelsSoCG}.
They showed that \OLP parameterized by height is XNLP-complete (implying that it is in XP, but $W[t]$-hard for every $t\ge 1$).
In contrast, \CLP is NP-hard even if restricted to instances of height~$4$, but it can be solved in polynomial time if restricted to instances of height at most~$3$.

\subparagraph{Other related work.}
Several other
restricted
variants of \textsc{Level Planarity} have been studied, e.g., 
\textsc{Clustered Level Planarity}~\cite{DBLP:conf/sofsem/ForsterB04,DBLP:journals/tcs/AngeliniLBFR15,DBLP:journals/talg/KlemzR19}, 
\textsc{T-Level Planarity}~\cite{DBLP:journals/dam/WotzlawSP12,DBLP:journals/tcs/AngeliniLBFR15,DBLP:journals/talg/KlemzR19}, 
and \textsc{Partial Level Planarity}~\cite{DBLP:conf/soda/BrucknerR17}.
In particular, in \textsc{Partial Level Planarity},
 a given level planar drawing of a subgraph of the input graph~$\mathcal G$ has to be extended to a
full 
drawing of~$\mathcal G$, 
which can be seen as a generalization of \OLP and, in the proper case, a specialization of \CLP.
  \textsc{Level Planarity} has been extended to surfaces different from the plane~\cite{DBLP:journals/jgaa/BachmaierBF05,DBLP:conf/gd/AngeliniLBFPR16,DBLP:conf/esa/BachmaierB08}.
There are also related problems with a more geometric flavor, e.g.,
finding a level planar straight-line drawing where each face is bounded by a convex polygon 
\cite{DBLP:journals/jda/HongN10,DBLP:conf/esa/Klemz21},
and problems where the input is an undirected graph without a level assignment and the task is to find a crossing-free drawing with y-monotone edges that, if interpreted as a level planar drawing,
satisfies or optimizes certain criteria, e.g., being proper or having minimum height \cite{DBLP:journals/algorithmica/BannisterDDEW19,DBLP:journals/algorithmica/DujmovicFKLMNRRWW08,DBLP:journals/siamcomp/HeathR92}.

\subparagraph{Contribution.}
As discussed above, the previous results of Br\"uckner and Rutter~\cite{DBLP:conf/soda/BrucknerR17} and Klemz and Rote~\cite{DBLP:journals/talg/KlemzR19} 
rule out the existence of FPT-time (even XP-time) algorithms for \CLP when considering the tree-width, path-width, tree-depth, or feedback vertex 
set number as a parameter, even when restricted to OLP or proper CLP instances (unless P=NP).
As all of these parameters are bounded\footnote{More precisely, $\mathrm{tw}(G)\le \mathrm{pw}(G)\le \mathrm{td}(G)-1\le\mathrm{vc}(G)$ and $\mathrm{fvs}(G)\le\mathrm{vc}(G)$.} by the vertex cover number, 
it is natural to study the parameterized complexity of \CLP with respect to this parameter.
We prove the following main result:

\begin{theorem}
\label{thm:fpt}
\CLP parameterized by the vertex cover number is FPT.
\end{theorem}

\noindent
In view of the previous intractability statements, \cref{thm:fpt} is best-possible in several regards:
a speed-up to polynomial time or a generalization to the aforementioned smaller graph parameters is not possible, even if restricting to OLP or proper CLP instances.

\subparagraph{Organization.}
The proof of \cref{thm:fpt} and the remainder of this paper are organized as follows.
We begin by introducing some basic notation, terminology, and other preliminaries in \cref{sec:types}.
In particular, we describe a partition of the vertex set of a given constrained level graph~$\mathcal G$ into different categories with respect 
to a given vertex cover~$C$ and we show that the vertices of two of these categories are, in some sense, easy to handle.
In \cref{sec:vis-core}, we introduce cores and (refined) visibility extensions of level planar drawings with respect to a fixed vertex cover $C$.
Intuitively, the core-induced subdrawing of a (refined) visibility extension of a constrained level planar drawing~$\Gamma^*$ of~$\mathcal G$ with respect 
to~$C$ is a drawing~$\Lambda_\mathrm{core}$ that captures crucial structural properties of~$\Gamma^*$ and whose total complexity is bounded in~$|C|$.
The latter allows us to efficiently obtain such a core-induced subdrawing via the process of exhaustive enumeration.
This is the first main step of the algorithm corresponding to the proof of \cref{thm:fpt}, which is described in \cref{sec:algorithm}.
Due to the properties of the core-induced subdrawing, it is then possible to place the remaining vertices 
in the subsequent main steps of the algorithm, each of which is 
concerned with the placement of the vertices of a particular vertex category.
We conclude with a discussion of an open problem in \cref{sec:conclusion}.
Proofs of statements marked with a (clickable) $\star$ can be found in the appendix.

\section{Preliminaries}
\label{sec:types}

\subparagraph{Conventions.}
Recall that in a level graph~$\mathcal G=(G=(V,E),\gamma)$, the graph~$G$ is directed by definition.
However, when it comes to vertex-adjacencies, we always refer to the underlying undirected graph of~$G$, that is, the neighborhood of $v\in V$ is $\mathrm{N}_G(v)=\{u\in V\mid (u,v)\in E$ $\vee (v,u)\in E\}$, the degree of $v$ is $|\mathrm{N}_G(v)|$, and ``a vertex cover of $G$'' refers to a vertex cover\footnote{A \emph{vertex cover} of an undirected graph $G=(V,E)$ is a vertex set $C\subseteq V$ such that every edge in $E$ is incident to at least one vertex in $C$. The \emph{vertex cover number} of $G$ is the size of a smallest vertex cover of~$G$.} of the underlying undirected graph of~$G$.
The \emph{level planar embedding} of a level planar drawing of~$\mathcal G$  lists, for each level $V_i$, the left-to-right sequence of vertices and edges intersected by the line $L_i$ in the drawing.
Note that this corresponds to an equivalence class of drawings from which an actual drawing is easily derived, which is why algorithms for constructing level planar drawings (including our algorithms) usually just determine a level planar embedding.
For brevity, we often use the term ``drawing'' as a synonym for ``embedding of a drawing''.

\subparagraph{Vertex categories \& notation.}
For $m\in \mathbb N$, we use $[m]$ to denote the set $\{1,2,\dots,m\}$.
Let~$\mathcal G=(G=(V,E),\gamma)$ be a (constrained) level graph and let~$C$ be a vertex cover of~$G$.
An \emph{ear} of~$\mathcal G$ with respect to~$C$ is a degree-2 vertex of $V\setminus C$ that is a source or sink.
For a subset $X \subseteq C$, we define $V_X(C) = \{v \in V\setminus C \mid \mathrm{N}_G(v) = X\}$.
We partition the vertices~$V\setminus C$ of the graph~$G$ into four sets $V_{=0}(C), V_{=1}(C), V_{=2}(C), V_{\ge 3}(C)$ where
$V_{=0}(C) = \{ v \in V\setminus C \mid \mathrm{deg}(v) = 0\}$,
$V_{=1}(C) = \{ v \in V\setminus C \mid \mathrm{deg}(v) = 1\}$ (the \emph{leaves}),
$V_{=2}(C) = \{ v \in V\setminus C \mid \mathrm{deg}(v) = 2\}$, and
$V_{\geq 3}(C) = \{ v \in V\setminus C \mid \mathrm{deg}(v) \geq 3\}$.
The set $V_{=2}(C)$ is further partitioned into two sets $V_{=2}^\mathrm e(C), V_{=2}^\mathrm t(C)$ where $V_{=2}^\mathrm e(C)$ contains the 
ears and $V_{=2}^\mathrm t(C)$ the non-ears, called \emph{transition} vertices.
Let~$v\in V_{=1}(C)$ and let~$c\in C$ denote its (unique) neighbor.
We say that~$v$ is a \emph{leaf} of~$c$.
Similarly, let~$u\in V_{=2}(C)$ and let~$c_a,c_b\in C$ denote its (unique) two neighbors.
We say that~$u$ is a \emph{transition vertex} (\emph{ear}) of~$c_a$ and~$c_b$ if~$u\in V_{=2}^\mathrm t(C)$ (if $u\in V_{=2}^\mathrm e(C))$.
We often omit $C$ if it is clear from the context.

Let $\mathcal G=(G=(V,E),\gamma,(\prec_i)_{i})$ be a constrained level graph and let~$C$ be a vertex cover of~$G$.
The main challenge when constructing a constrained level planar drawing of~$\mathcal G$ is the placement of the leaves, ears, and transition vertices (along with their incident edges).
Indeed, it is not difficult to insert the isolated vertices (which include $V_{=0}$) in a post-processing step (performing a topological sort on each level), see \cref{lem:isolated}.
Moreover, since we may assume~$G$ to be planar, the size of $V_{\geq 3}(C)$ is linear in~$|C|$.
This well known bound can be derived, e.g., by combining the fact that the complement of a vertex cover is an independent set with the following statement (setting $A=C$).

\begin{lemma}[{\hspace{1sp}\cite[Corollary 9.25]{DBLP:books/sp/CyganFKLMPPS15}}]\label{lem:vc-machinery}
Let~$G=(V,E)$ be a planar graph and $A\subseteq V$.
Then there are at most $2|A|$ connected components in the subgraph of~$G$ induced by $V\setminus A$ that are adjacent to more than two vertices of~$A$.\qed
\end{lemma}

\begin{corollary}[Folklore]
\label{lem:xg2:Vgeq3InKexp3}
Let $G$ be a planar graph and let $C$ be a vertex cover of $G$. Then $|V_{\geq 3}(C)|\le 2k$, where $k=|C|$.\qed
\end{corollary}

\begin{restatable}[\restateref{lem:isolated}]{lemma}{Isolated}
\label{lem:isolated}
 Let~$\mathcal{G} = (G, \gamma, (\prec_i)_i)$ be a constrained level graph, let~$G'$ be the subgraph of~$G$ induced by the non-isolated vertices $V'$,
and let~$\gamma'$ and~$(\prec'_i)_i$ be the restrictions of~$\gamma$ and~$(\prec_i)_i$ to~$V'$, respectively.
There is an algorithm that, given~$\mathcal G$ and a constrained level planar drawing~$\Gamma'$ of $\mathcal{G}'=(G',\gamma',(\prec'_i)_i)$, constructs a constrained level-planar drawing of~$\mathcal G$ in polynomial time.
\end{restatable}

\noindent Our main algorithm will exploit the fact that only few ears may share a common level:

\begin{restatable}[\restateref{cl:2k-ears-per-level}]{lemma}{TwoK}
	\label{cl:2k-ears-per-level}
	Let $\mathcal G=(G,\gamma)$ be a level graph, let $\Gamma$ be a level planar drawing of~$\mathcal G$, let $C$ be a vertex cover of~$G$.
	Then there are at most $2|C|$ ears with respect to~$C$ per level. 
\end{restatable}

\subparagraph{Compatible edge orderings.}
Let~$\Gamma$ be a level planar drawing of a (possibly constrained) level graph~$\mathcal G=(G=(V,E),\gamma)$ without isolated vertices.
We will now define a useful (not necessarily unique) linear order~$\prec^{\mathrm{e}}$ on the edges~$E$ with respect to~$\Gamma$.
We refer to~$\prec^{\mathrm{e}}$ as an edge ordering of~$E$ that is \emph{compatible} with~$\Gamma$.
Compatible edge orderings can be seen as a generalization of a linear order described in~\cite[Proof of Lemma~4.4]{DBLP:journals/talg/KlemzR19} for a 
set of pairwise disjoint y-monotone paths, which in turn follows considerations about horizontal separability of y-monotone sets by 
translations~\cite{deB,amr-qpscu-15,gy-tsr-80,gy-tsr-83}.
Intuitively, $\prec^{\mathrm{e}}$ is a linear extension of a partial order in which $e\in E$ precedes $f\in E$ if it is possible to shoot a horizontal rightwards ray from $e$ to $f$ in $\Gamma$ without crossing other edges before reaching~$f$.
Formally, 
we say that a vertex~$v$ is \emph{visible from the left} in~$\Gamma$ if the horizontal ray~$r_v$ emanating from~$v$ to the left intersects~$\Gamma$ only in~$v$.
We say that an edge~$e=(u,v)$ is \emph{visible from the left} in~$\Gamma$ if 
the closed (unbounded) region that is to the left of $e$ and whose boundary is described by $e, r_u, r_v$ intersects~$\Gamma$ only in~$e$.
The order~$\prec^{\mathrm{e}}$ is now constructed as follows:
the minimum of~$\prec^{\mathrm{e}}$ is an edge~$e_1$ of~$E$ that is visible from the left in~$\Gamma$.
Such an edge always exists~\cite{DBLP:journals/talg/KlemzR19,gy-tsr-80,gy-tsr-83}:
among the edges whose lower endpoint is visible from the left, the edge with the topmost lower endpoint is visible from the left.
Let~$\Gamma'$ denote the drawing derived from~$\Gamma$ by removing~$e_1$ and any isolated vertices created by the removal of~$e_1$.
The restriction of~$\prec^{\mathrm{e}}$ to the remaining edges~$E\setminus e_1$ corresponds to an edge ordering compatible with~$\Gamma'$, which is constructed recursively.
Note that $\mathcal G$ and~$\prec^{\mathrm{e}}$ uniquely describe
the drawing~$\Gamma$ and,
given $\mathcal G$ and $\prec^{\mathrm{e}}$, it is possible to construct~$\Gamma$ in polynomial time (by traversing $\prec^{\mathrm{e}}$ in reverse).

\section{Visibility extensions and cores}
\label{sec:vis-core}

In this section, 
we introduce and study 
(refined) visibility extensions and cores of level planar drawings.
We will see that the core-induced subdrawing of a (refined) visibility extension of a level planar drawing~$\Gamma$ with respect to some 
vertex cover~$C$ captures crucial structural properties of~$\Gamma$ while having a size that is bounded in~$|C|$. 

\subparagraph{Visibility extensions.}
Let~$\Gamma$ be a level planar drawing of a level graph $\mathcal G=(G=(V,E),\gamma)$.
A \emph{visibility edge}~$e$ for~$\Gamma$ is
(1) a y-monotone curve that joins two vertices of~$\Gamma$ and can be inserted into~$\Gamma$ without creating any crossings (but possibly a pair of parallel edges);
or (2) a horizontal segment that joins two consecutive vertices on a common level of~$\Gamma$ and can be inserted into~$\Gamma$ without creating any crossings.
A \emph{visibility extension} of~$\Gamma$ with respect to a vertex set~$V'\subseteq V$
is a drawing~$\Lambda$ derived from~$\Gamma$ by inserting a maximal set of 
pairwise non-crossing visibility edges incident only to vertices of~$V'$ such that for each pair~$e,e'$ of parallel edges in~$\Lambda$ there is at least one 
vertex of~$V'$ in the interior of the simple closed curve formed by~$e$ and~$e'$; for an illustration see \cref{fig:core}(a).
We remark that if $V'=V$, then $\Lambda$ is essentially an interior triangulation containing~$\Gamma$. However, we will always choose $V'$ to be a (small) vertex cover, resulting in a much sparser yet still connected augmentation:

\begin{restatable}[\restateref{lem:VC-connected}]{lemma}{Cconnected}
Let $\mathcal G=(G=(V,E),\gamma)$ be a level graph without isolated vertices, let~$C$ be a vertex cover of~$G$, let~$\Gamma$ be a level 
planar drawing of~$\mathcal G$,  let~$\Lambda$ be a visibility extension of~$\Gamma$ with respect to~$C$, and let~$\Lambda_C$ be the subdrawing 
of~$\Lambda$ that is induced by~$C$.
Then~$\Lambda_C$ is connected and has~$\mathcal O(k)$ edges, where $k=|C|$.
\label{lem:VC-connected}
\end{restatable}

\subparagraph{Cores and refined visibility extensions.}
Intuitively, the core of a level planar drawing is a subset of the vertex set with certain crucial properties.
To define it formally, we will first classify the ears of the drawing according to several categories.
The concepts introduced in this paragraph are illustrated in \cref{fig:core}(b).
Let $\mathcal G=(G=(V,E),\gamma)$ be a level graph, let~$C$ be a vertex cover of~$G$, and let~$\Gamma$ be a level 
planar drawing of~$\mathcal G$.
Consider an ear $v\in V_{=2}^\mathrm e(C)$  with neighbors $c_a,c_b$ where $\gamma(c_a)\ge \gamma(c_b)$.
If $\gamma(v)>\gamma(c_a)$, we say that~$v$ is a \emph{top} ear.
Otherwise (if $\gamma(v)<\gamma(c_b)$), we say that~$v$ is a \emph{bottom} ear.
Assume that $\gamma(c_a)> \gamma(c_b)$ and that~$v$ is a top ear.
If in~$\Gamma$ the edge~$c_bv$ is drawn to the left (right) of~$c_a$, we say that~$v$ is a \emph{left} (\emph{right}) ear in~$\Gamma$.
The terms ``left'' and ``right'' are defined analogously for bottom ears.
If $\gamma(c_a)=\gamma(c_b)$, we consider $v$ to be a left ear if it is a top ear; otherwise it is a right ear.
Consider a pair $c_a,c_b\in C$ with at least one left ear in~$\Gamma$ and let~$\Gamma'$ denote the subdrawing of~$\Gamma$ induced by the set of edges 
that are incident to at least one left ear of $c_a,c_b$ in~$\Gamma$.
Note that either all these ears are top ears or all these ears are bottom ears in~$\Gamma$, and they are arranged in a nested fashion.
In case of top (bottom) ears, we refer to the unique one with the largest (smallest) y-coordinate as the \emph{outermost} left ear of $c_a,c_b$.
The \emph{innermost} left ear of $c_a,c_b$ is defined symmetrically.
If~$\Gamma'$ has an interior (i.e., bounded) face~$f$
such that the open region enclosed by the boundary of~$f$ contains a vertex of~$C$ in~$\Gamma$,
then we say that the two 
ears~$v_1,v_2$ on the boundary of~$f$ are \emph{bounding ears} of~$c_a,c_b$ in $\Gamma$ with respect to~$C$.
Moreover, we say that~$v_1,v_2$ are a pair of \emph{matching} bounding ears whose \emph{region} corresponds to~$f$.
The terms ``outermost'', ``innermost'' and ``(region of matching pair of) bounding ears'' are analogously defined for the right ears of $c_a,c_b$.
Every vertex of~$\Gamma$ that is an outermost, innermost, or bounding ear (with respect to some pair $c_a,c_b\in C$) is called \emph{crucial} with respect to~$C$.

The \emph{core} of $\Gamma$ with respect to~$C$ is the (unique) subset of~$V$ that contains $C$, $V_{\ge 3}(C)$, as well as all crucial ears of~$\Gamma$ with respect to~$C$.
The subdrawing~$\Lambda_\mathrm{core}$ of a visibility extension~$\Lambda$ of~$\Gamma$ with respect~$C$ that is induced by the core of $\Lambda$ with respect to $C$ captures crucial structural properties of~$\Gamma$, which we will exploit in our main algorithm for constructing constrained level planar drawings in FPT-time.
Due to the fact that~$\Lambda_\mathrm{core}$ has only $\mathcal O(|C|)$ vertices and edges, it is not difficult to ``guess''~$\Lambda_\mathrm{core}$ in \emph{XP}-time (via the process of exhaustive enumeration) when given $\mathcal G$ and $C$.
The main bottleneck is the enumeration of all possible sets of crucial ears.
To improve the runtime of this step to \emph{FPT}-time,
we will now describe a variant of visibility extensions that contains some additional ears, which take over the role of the original crucial vertices. Loosely speaking, one can create such a drawing by placing one or two new ears near each crucial ear in a visibility extension. The resulting augmentation retains the helpful structural properties of its underlying visibility extension and we will see that the (positions of the) crucial vertices of some such augmentation can be guessed more efficiently since we can restrict the possible levels of these new vertices to a small set.
Formally, a \emph{refined visibilty extension}~$\Lambda'$ of~$\Gamma$ is a crossing-free drawing of a level graph~$\mathcal G'=(G'=(V', E'),\gamma')$ such that~$G$ is a subgraph of~$G'$, $C$ is a vertex cover of~$G'$,
every vertex in~$V'\setminus V$ is an ear with respect to~$C$ and its incident edges are drawn as y-monotone curves,
the subdrawing of $\Lambda'$ induced by~$V$ is a visibility extension~$\Lambda$ of~$\Gamma$, 
the crucial ears of~$\Lambda'$ are precisely the vertices in $V'\setminus V$,
and~$|V'\setminus V| \in \mathcal O(|C|)$.

\begin{figure}[tbh]
    \centering
    \includegraphics[]{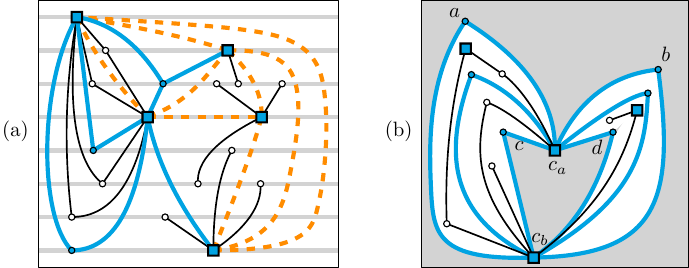}
    \caption{
    In this (and all other) figure(s), filled square vertices belong to a vertex cover $C$ of the depicted graph and
    filled (round or square) vertices belong to the core of the shown drawing with respect to~$C$.
    (a) A drawing $\Gamma$ (non-dashed edges) is augmented with visibility edges (dashed) to obtain a visbility extension~$\Lambda$ with respect to $C$ (note that this augmentation is not unique).
    The thick (non-dashed or dashed) edges and filled vertices represent $\Lambda_\mathrm{core}$.
    (b) All filled round vertices are top crucial ears of $c_a$ and $c_b$.
    All of them are bounding ears except for $b$ and $c$.
    The vertex $a$ / $b$ / $c$ / $d$ is an outermost left / outermost right / innermost left / innermost right ear.
	}
    \label{fig:core}
\end{figure}

\begin{restatable}[\restateref{lem:size-of-core}]{lemma}{SizeCore}
Let $\mathcal G=(G=(V,E),\gamma)$ be a level graph without isolated vertices, let~$C$ be a vertex cover of~$G$, let~$\Gamma$ be a level 
planar drawing of~$\mathcal G$, let~$\Lambda$ be a refined or non-refined visibility extension of~$\Gamma$ with respect to~$C$, and let~$\Lambda_{\mathrm core}$ the 
subdrawing of~$\Lambda$ induced by the core of~$\Lambda$ with respect to~$C$.
Then~$\Lambda_{\mathrm{core}}$ is connected and has $\mathcal O(k)$ vertices and $\mathcal O(k)$ edges, where $k=|C|$.
\label{lem:size-of-core}
\end{restatable}

\section{Algorithm}
\label{sec:algorithm}

In this section, we describe the algorithm corresponding to the proof of \cref{thm:fpt}.
Let $\mathcal G=(G=(V,E),\gamma,(\prec_i)_i)$ be a constrained level graph and let~$C$ be a vertex cover of~$G$.
Our goal is to construct a constrained level planar drawing of~$\mathcal G$ or correctly report that such a drawing does not exist.
In view of \cref{lem:isolated}, we may assume that~$\mathcal G$ has no isolated vertices.
To construct the desired drawing, we proceed in three main steps.
In Step~1,
 we ``guess'' a core-induced subdrawing of a refined visibility extension of a constrained level planar drawing of~$\mathcal G$ with 
respect to~$C$ (via the process of exhaustive enumeration).
In Step~2,
  we augment our drawing by inserting the transition vertices of~$\mathcal G$ with respect to~$C$.
In Step~3,
  we finalize our drawing by inserting the leaves and ears of~$\mathcal G$ with respect to~$C$.

\subparagraph{Step 1: Guessing a core-induced subdrawing.}
Assume that there is a constrained level planar drawing~$\Gamma^*$ of~$\mathcal G$, let~$\Lambda^*$ be a refined visibility extension of~$\Gamma^*$ with 
respect to~$C$, and let~$\Lambda_{\mathrm{core}}$ be the subdrawing of~$\Lambda^*$ induced by the core of~$\Lambda^*$ with respect to~$C$.
The procedures corresponding to Steps~2 and~3 of our algorithm are guaranteed to produce a constrained level planar drawing (not necessarily~$\Gamma^*$) 
of~$\mathcal G$ when given $\Lambda_{\mathrm{core}}$.
Hence, the goal of Step~1 is to determine (or, rather, guess) $\Lambda_{\mathrm{core}}$, given~$\mathcal G$ and~$C$.
More precisely, we will construct a set~$\mathcal F$ of drawings such that $\Lambda_{\mathrm{core}}\in \mathcal F$ and the number 
$|\mathcal F|$ of drawings in~$\mathcal F$ is sufficiently small.
For each drawing in~$\mathcal F$, we then apply Steps~2 and~3 of the algorithm (incurring a factor of $|\mathcal F|$ in the total running time).
Given that $\Lambda_{\mathrm{core}}\in \mathcal F$, one of the iterations is guaranteed to terminate with a constrained level planar drawing of~$\mathcal G$.

\begin{restatable}[\restateref{lem:step1}]{lemma}{StepOne}
Let $\mathcal G=(G=(V,E),\gamma,(\prec_i)_i)$ be a constrained level graph without isolated vertices, let~$C$ be a vertex cover of~$G$, and let~$\Gamma^*$ be a constrained level 
planar drawing of~$\mathcal G$.
There is an algorithm that, given~$\mathcal G$ and~$C$, 
constructs a set~$\mathcal F$ of  
$2^{\mathcal O(k\log k)}$
drawings in  
$2^{\mathcal O(k\log k)}\cdot n^{\mathcal O(1)}$ 
time, where
$n=|V|$ and $k=|C|$,
such that all drawings in~$\mathcal F$ have size~$\mathcal O(k)$ and are level planar drawings of subgraphs of~$G$ induced by
$C$ and $V_{\ge 3}(C)$ 
that respect $\gamma$ and the orderings $\prec_i$
and are augmented by some visibility edges and additional ears (with respect to $C$).
Further, there exists a refined visibility extension $\Lambda^*$ of $\Gamma^*$ such that the subdrawing~$\Lambda_{\mathrm{core}}$ of $\Lambda^*$ induced by the core of~$\Lambda^*$ with respect to~$C$ is contained in~$\mathcal{F}$.
\label{lem:step1}
\end{restatable}

\begin{proof}[Proof sketch]
We introduce the following terminology:
let $x$ be a vertex in a level planar drawing (possibly augmented by some horizontal edges).
Let $\ell$ be the y-coordinate of~$x$ and let~$\ell'$ be the largest y-coordinate of a vertex below $x$ (if there no such vertex, we set~$\ell'=\ell-1$).
We say that the line $L_{(\ell+\ell')/2}$ is \emph{directly below} $x$.
The line \emph{directly above} $x$ is defined symmetrically.

We proceed in two main steps.
In the first main step, we show that there \emph{exists} a refined visibility extension~$\Lambda^*$ of~$\Gamma^*$.
To this end, we start with a visibility extension~$\Lambda$ of~$\Gamma^*$ and describe an incremental strategy that performs a total of~$\mathcal O(k)$ augmentation steps, in each of which a new ear is added that takes over the role of a crucial ear in~$\Gamma^*$.
(The description of this first main step is deferred to the appendix.)
In the second main step, we discuss the construction of the desired family~$\mathcal F$.
To this end, let~$\Lambda_{\mathrm{core}}$ be the subdrawing of $\Lambda^*$ induced by the core of~$\Lambda^*$ with respect to~$C$.
The drawing $\Lambda_{\mathrm{core}}$ is uniquely described by $\mathcal G$, $C$, the set of visibility edges of $\Lambda_{\mathrm{core}}$ (and $\Lambda^*$), the set of crucial ears of $\Lambda_{\mathrm{core}}$ (and $\Lambda^*$) together with their level assignments and their incident edges, and a compatible edge ordering of the nonhorizontal edges of $\Lambda_{\mathrm{core}}$.
The graph~$\mathcal G$, as well as the vertex cover $C$ are given, so it suffices to enumerate all possible options for the remaining elements.

There are $m_{\mathrm{vis}}\in \mathcal O(k)$ visibility edges by \cref{lem:size-of-core}
and each of these visibility edges joins a pair of vertices in $C$.
Hence, there are at most $\binom{k}{2}^{m_{\mathrm{vis}}}\subseteq k^{\mathcal O(k)}\subseteq 2^{\mathcal O(k\log k)}$ possible options for choosing the set of visibility edges.
To enumerate the set of crucial ears along with their level assignments,
we mimic the aforementioned incremental strategy for constructing~$\Lambda^*$:
we first enumerate all options to pick the pair of neighbors of the first new vertex along with its level,
then, for each of these options, we enumerate all options to pick the pair of neighbors of the the second vertex along with its level, etc.,
until we have obtained all options to pick the desired $\mathcal O(k)$ vertices together with their levels.
More precisely, suppose we have already enumerated all options to pick the first~$i$ vertices together with their neighbors and levels.
For each of these options, to enumerate all options to pick the next vertex~$v'$, we go through all ways to pick its two neighbors~$u,w\in C$ and through all ways to pick the level of~$v'$.
There are $\mathcal O(k)$ pairs of vertices in~$C$ with ears (by \cref{lem:size-of-core}).
To bound the number of ways to pick the level of~$v'$, we make use of the fact that whenever the incremental strategy for constructing~$\Lambda^*$ places a new vertex~$v'$, it is assigned to a new level directly above or below a level of one of the following categories:
\begin{itemize}
\item a level with a vertex in $C$ ($\mathcal O(k)$ possibilities),
\item a level with a vertex in $V_{\ge 3}(C)$ ($\mathcal O(k)$ possibilities by \cref{lem:xg2:Vgeq3InKexp3}),
\item a level of a vertex that does not belong to~$\mathcal G$, i.e., a level used for
one of the already placed vertices ($\mathcal O(i)\subseteq \mathcal O(k)$ possibilities),
\item a level with a top-most or bottom-most vertex of $\mathcal G$ ($\mathcal O(1)$ possibilities),
\item a level with a top-most top ear, a top-most bottom ear, a bottom-most top ear, or a bottom-most bottom ear of some pair of vertices in~$C$ ($\mathcal O(k)$ possibilities by \cref{lem:size-of-core}),
\item a level with a top-most or bottom-most vertex of a connected component that contains a vertex of~$C$ in the graph obtained by removing $u$ and $w$ from the current graph ($\mathcal G$ together with the visbility edges and the already added vertices) ($\mathcal O(k)$ possibilities).
\end{itemize}
In total, for a fixed pair of neighbors~$u,w$, there are thus $\mathcal O(k)$ options to pick a level for~$v'$.
We immediately discard level assignments for which~$v'$ is no ear.
By multiplying with the number of ways to choose the neighbors, we obtain~$\mathcal O(k^2)$ options to choose~$v'$ and its level.
Multiplying the number of options for all $\mathcal O(k)$ steps together, we obtain a total number of $k^{\mathcal O(k)}\subseteq 2^{\mathcal O(k\log k)}$ ways to create the set of crucial ears along with their levels.
By multiplying with the number of ways to choose the visibility edges, we obtain a total of $2^{\mathcal O(k\log k)}$ options to choose the graph that corresponds to~$\Lambda_{\mathrm{core}}$.
For each of these options we enumerate all $k^{\mathcal O(k)}\subseteq 2^{\mathcal O(k\log k)}$ permutations of the set of non-horizontal edges and, interpreting the permutation as a compatible edge order, try to construct a level planar drawing for which this order is compatible (cf.\ \cref{sec:types}).
If we succeed,
we check whether the drawing is conform with $(\prec_i)_i$ and can be augmented with the horizontal visibility edges. If so,
we include the drawing in the set~$\mathcal F$ of reported drawings.
The size of the thereby constructed set~$\mathcal F$ is bounded by~$2^{\mathcal O(k\log k)}$ and it is guaranteed to contain $\Lambda_{\mathrm{core}}$ by construction.
\end{proof}

\subparagraph{Step 2: Inserting transition vertices.}
We now describe how to insert the transition vertices into the core-induced subdrawing~$\Lambda_\mathrm{core}$ of the (refined) visibility extension~$\Lambda^*$.
Our plan is to first show that in~$\Lambda^*$, every transition vertex is placed ``very close  to'' some visibility edge.
Intuitively, this means that the visibility edges of~$\Lambda_\mathrm{core}$ act as placeholders near which the transition vertices have to be placed.
We will describe a procedure that does so while carefully taking into account the given partial orderings~$\prec_i$ and prove its correctness by means of an (somewhat technical) exchange argument.
To formalize the notion of ``very close to'', 
let $e$ be an edge of a level planar drawing joining two vertices $a,b$
such that there is a degree-2 vertex $t$ with neighbors $a,b$ and
$\gamma'(b)<\gamma'(t)<\gamma'(a)$
where $\gamma'$ is the level assignment.
We say that~$t$ is drawn in the \emph{vicinity} of~$e$ with respect to a vertex set~$V'$ if the simple closed curve formed by~$e$ and the two edges incident to~$t$ does not contain a vertex of~$V'$ in its interior.

\begin{restatable}[\restateref{lem:insertTransition}]{lemma}{StepTwo}
\label{lem:insertTransition}
Let $\mathcal G=(G=(V,E),\gamma,(\prec_i)_i)$ be a constrained level graph without isolated vertices, let~$C$ be a vertex cover of~$G$, let~$\Gamma^*$ be a constrained level 
planar drawing of~$\mathcal G$, let~$\Lambda^*$ be a refined or non-refined visibility extension of~$\Gamma^*$ with respect to~$C$, and let~$\Lambda_{\mathrm core}$ the subdrawing
of~$\Lambda^*$ induced by the core of~$\Lambda^*$ with respect to~$C$.
There is an algorithm that, given~$\mathcal G$, $C$ and~$\Lambda_{\mathrm core}$,
inserts all transition vertices~$V_{=2}^{\mathrm t}(C)$ (and their incident edges) into vicinities (with respect to~$C$) of visibility edges in~$\Lambda_{\mathrm {core}}$ in polynomial time such that the resulting 
drawing $\Lambda^\mathrm{t}_{\mathrm {core}}$ can be 
extended to a drawing whose restriction to~$G$ is a constrained level planar drawing of~$\mathcal G$.
\end{restatable}

\subparagraph{Step 3: Inserting leaves and ears.}

In this step, we start with the output of Step~2 and finalize our drawing by placing all the vertices that are still missing.

\begin{lemma}
\label{lem:insertLeftovers}
Let $\mathcal G=(G=(V,E),\gamma,(\prec_i)_i)$ be a constrained level graph without isolated vertices, let~$C$ be a vertex cover of~$G$, let~$\Gamma^*$ be a constrained level 
planar drawing of~$\mathcal G$, let~$\Lambda^*$ be a refined or non-refined visibility extension of~$\Gamma^*$ with respect to~$C$
in which each transition vertex is placed in the vicinity of some visibility edge with respect to $C$,
and let $\Lambda_{\mathrm {core}}$ ($\Lambda^\mathrm{t}_{\mathrm {core}}$) be the subdrawing of $\Lambda^*$ induced by the core (and the transition vertices) of~$\Lambda^*$ with respect to~$C$.
There is an algorithm that, given~$\mathcal G$, $C$, $\Lambda_{\mathrm {core}}$ and~$\Lambda^\mathrm{t}_{\mathrm {core}}$, extends $\Lambda_{\mathrm {core}}$ to a drawing~$\Gamma$ whose restriction to $G$ is a constrained level planar drawing of~$\mathcal G$ in $2^{\mathcal O(k^2\log k)}n^{\mathcal O(1)}$ time, where $n=|V|,k=|C|$.
\end{lemma}

\begin{proof}
The only vertices of $G$ that are missing in $\Lambda^\mathrm{t}_{\mathrm {core}}$ are the leaves and the non-crucial ears with respect to $C$ (in case $\Lambda^*$ is a \emph{refined} visibility extension, the non-crucial ears are exactly the ears of~$\mathcal G$).
Our plan to insert them into our drawing is as follows.
We begin by introducing more structure in $\Lambda^\mathrm{t}_{\mathrm {core}}$ and $\Lambda^*$ by adding some additional visibility edges and making some normalizing assumptions, which will simplify the description of the upcoming steps.
In particular, this step will ensure that for each missing ear, there are essentially only (up to) two possible placements, which will allow us to enumerate all possible ear placements (so-called ear orientations) on a given level in FPT-time.
We then describe a partition of the plane into so-called cells in a way that is very reminiscent of the well-known trapezoidal decomposition from the field of computational geometry,
cf.~\cite{DBLP:books/lib/BergCKO08}.
We merge some cells into so-called channels, which correspond to connected y-monotone regions in which the missing leaves along with their incident edges will be drawn (a region is \emph{y-monotone} if its intersection with every horizontal line is connected).
We then introduce (and describe an enumerative process that constructs in FPT-time) a so-called traversal sequence that is compatible with~$\Lambda^*$, which is a sequence of sets of channels with several useful structural properties related to~$\Lambda^*$.
In particular, this sequence, in some sense, sweeps the plane from left to right in a way where for each edge incident to a leaf in~$\Lambda^*$,
at some point there is a channel that contains it.
Exploiting the properties of the traversal sequence, we then describe how to construct a so-called insertion sequence for the leaves on a given level with respect to a given placement of the ears of that level in polynomial time.
Such an insertion sequence does not necessarily exist for every placement of ears, but we are guaranteed to find one by enumerating all possible ear placements of the level.
This computation is performed independently for each level.
Finally, we show how to construct in polynomial time the desired drawing~$\Gamma$ when given an insertion sequence
along with its ear placement for each level.
Notably, the final step can be executed even if  some of the ear placements are different from the ones used in~$\Lambda^*$.
Let us proceed to formalize these ideas.

\subparagraph{Augmenting and normalizing $\Lambda_{\mathrm {core}}$, $\Lambda^\mathrm{t}_{\mathrm {core}}$, and $\Lambda^*$.}
Let $e = (u,v)$ be a visibility edge of $\Lambda^*$ (and $\Lambda^\mathrm{t}_{\mathrm {core}}$) that has at least one transition vertex in its vicinity.
In both $\Lambda^\mathrm{t}_{\mathrm {core}}$ and $\Lambda^*$, we add two copies of~$e$; one directly to the left of the leftmost transition vertex and the other 
directly to the right of the rightmost transition vertex in the vicinity of $e$,
which is possible to do in a y-monotone fashion and without introducing any crossings;
see \cref{fig:cells}(a).
Note that the region enclosed by these two edges only contains transitions vertices of $u$ and $v$, as well as leaves of $u$ or $v$.
We repeat this operation for all visibility edges~$e$.

The following steps are illustrated in \cref{fig:cells}(b).
Let $f$ be a face of $\Lambda^{t}_\mathrm{core}$  that is bounded by four vertices $v_1, v_2, u_1, u_2$ where $v_1, v_2 \in C$ and $u_1, u_2 \in V_{\{v_1, v_2\}}$ are either both left ears or both right ears. 
Without loss of generality, assume they are both left ears with
$\gamma(v_1)\le \gamma(v_2)<\gamma(u_1)<\gamma(u_2)$.
We add copies of the edges $(v_1,u_2)$ and $(v_2,u_2)$ in~$f$ in both $\Lambda^\mathrm{t}_{\mathrm {core}}$ and $\Lambda^*$,
which can be done without introducing crossings.
These edges partition~$f$ into three regions.
Note that in~$\Lambda^*$ these regions only contain leaves and non-crucial ears.
Without loss of generality, we will assume that each leaf in~$f$ is either placed in the region bounded by $(v_1,u_2)$ and its copy or the region bounded by $(v_2,u_2)$ and its copy
(note that a leaf $v$ that is adjacent to $v_1$ cannot have a constraint
of the form $w\prec_i v$, where $w$ is a non-crucial ear in $f$;
the situation for leaves adjacent to $v_2$ is symmetric).
Thus, the remaining (central) third region only contains non-crucial ears
and is henceforth called an \emph{ear-face} of $\Lambda^\mathrm{t}_\mathrm{core}$.
We repeat this modification for all faces such as~$f$.

For the remainder of the proof, $\Lambda^\mathrm{t}_{\mathrm {core}}$ and $\Lambda^*$ are used to refer to the thusly augmented and normalized drawings.
We also add all the new edges to $\Lambda_{\mathrm {core}}$ and use $\Lambda_{\mathrm {core}}$ to refer to this augmentation.
Note that this implies that it now suffices to search for a drawing of~$\mathcal G$ in which every non-crucial ear is placed in an ear-face, whereas no leaf is placed in an ear-face.

\begin{figure}[tbh]
    \centering
    \includegraphics[]{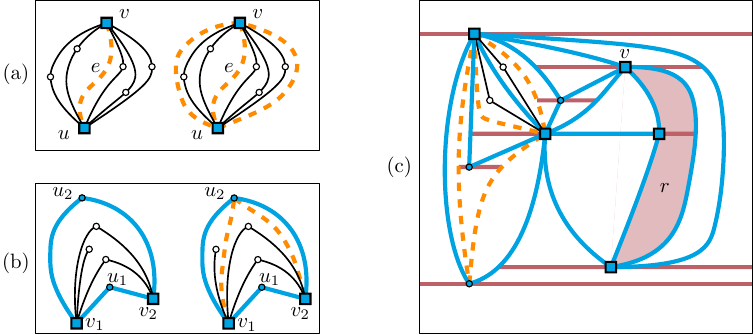}
    \caption{
    Like in all figures, filled square vertices belong to a vertex cover $C$ of the depicted graph and
    filled (round or square) vertices belong to the core of the shown drawing with respect to~$C$.
    Figures (a) and (b) visualize how $\Lambda_{\mathrm {core}}$, $\Lambda^\mathrm{t}_{\mathrm {core}}$, and $\Lambda^*$ are augmented by (a) enclosing transition vertices and (b) creating ear-faces with visibility edges (dashed). Moreover, (b) illustrates our normalizing assumption, i.e., the leaf can be moved to the exterior of the ear-face without violating any constraints.
    (c) The drawing $\Lambda^\mathrm{t}_{\mathrm {core}}$ with its additional visibility edges (dashed) from the augmentation step and the horizontal rays and segments (thick, red) from the cell decomposition.
	The shaded region corresponds to a channel $(v, r, R)$ from $v\in C$ to the cell $r$ with $|R| = 2$.
	}
    \label{fig:cells}
\end{figure}

\subparagraph{Decomposition into cells.}
We will now describe a partition of the plane that essentially corresponds to a
trapezoidal decomposition (cf.~\cite{DBLP:books/lib/BergCKO08}) of $\Lambda_{\mathrm {core}}$;
for illustrations refer to \cref{fig:cells}(c):
for each vertex $v$ in $\Lambda_{\mathrm {core}}$, shoot a horizontal ray from $v$ to the left until hitting an edge or vertex of $\Lambda_{\mathrm {core}}$, then add the corresponding segment to $\Lambda_{\mathrm {core}}$.
In case the ray does not intersect any part of $\Lambda_{\mathrm {core}}$, add the ray itself to $\Lambda_{\mathrm {core}}$.
Perform a symmetric augmentation by shooting a horizontal ray from $v$ to the right.
The maximal connected regions of the resulting partition of the plane are henceforth called \emph{cells}.
We consider the cells to be closed.
Note that each cell is y-monotone and bounded by up to two horizontal segments or rays and up to two y-monotone curves.
By \cref{lem:size-of-core}, $\Lambda_{\mathrm {core}}$ has $\mathcal O(k)$ vertices and edges (note that the augmentation step copies each edge at most twice) and, hence, it has $\mathcal O(k)$ faces.
Consequently, the number of cells is also $\mathcal O(k)$ since the insertion of a single segment or ray can only increase the number of faces (or, rather, maximal connected regions) by one.

\subparagraph{Channels.}
Let $v \in C$, let $R$ be a set of cells that do not belong to ear-faces, and let $r\in R$.
Further, assume that $R$ contains a cell that is incident to~$v$. 
The triple $c=(v,r,R)$ is a \emph{channel from $v$ to $r$}
 if it is possible to draw a y-monotone curve in $\Lambda^\mathrm{t}_{\mathrm {core}}$ in the interior of the union of $R$ that intersects  each cell in $R$ and does not cross any edge of $\Lambda^\mathrm{t}_{\mathrm {core}}$;
as illustrated in \cref{fig:cells}(c).
We say that $c$ \emph{can be used} by a leaf $w \in V_{\{v\}}$ and that the edge~$e_w$ incident to $w$ \emph{can be drawn in} $c$ if~$e_w$ can be drawn in $\Lambda^\mathrm{t}_{\mathrm {core}}$ in the union of $R$ without any crossings and there is no channel $(v,r',R')$ with this property for which $R'\subset R$.
Further, we say that $c$ \emph{is used} by $w$ if it can be used by $w$
and the edge incident to $w$ is drawn in the union of $R$ in $\Lambda^*$.
We use $U$ to denote the set of all channels and $U_\mathrm{used}\subseteq U$ to denote the set of all channels that are used.
The connectivity of $\Lambda^\mathrm{t}_{\mathrm {core}}$ (cf.\ \cref{lem:size-of-core}) can be used to show:

\begin{restatable}[\restateref{claim:channelnumber_bounded}]{claim}{ChannelNumber}
	\label{claim:channelnumber_bounded}
	$|U_{\mathrm{used}}|\le |U| \in \mathcal O(k^2)$.
\end{restatable}

\subparagraph{Traversal sequences.}
Let $\mathcal{U} =  (\mathcal{U}_1, \mathcal{U}_2, \dots, \mathcal{U}_m)$ be a sequence of sets of channels.
We say~$\mathcal{U}$ is a \emph{traversal sequence} if the following properties are fulfilled (see \cref{fig:traversal} for illustrations):
\begin{enumerate}[leftmargin=*,label={(T\arabic*)}]
	\item \label{T:unique_y_coord}
		Let $i\in [m]$ and let $(v,r,R)\in \mathcal{U}_i$ and $(v',r',R')\in\mathcal{U}_i$
		with $v\neq v'$.
		Then the intersection of the line $L_{\gamma(v)}$ with the interior of the union of $R'$ is empty.
	\item  \label{T:interval}
		Let $c$ be a channel and let $a \leq b$ be two indices such that $c \in \mathcal{U}_a$ and $c \in \mathcal{U}_b$. 
		Then for every $a \leq i \leq b$, $c \in \mathcal{U}_i$. 
\end{enumerate}
We say a channel $u$ is \emph{active} in $\mathcal{U}_i$ if it is contained in it, and otherwise it is \emph{inactive} in it. 
We say a traversal sequence $\mathcal U=(\mathcal{U}_1, \mathcal{U}_2, \dots, \mathcal{U}_m)$ is \emph{compatible} with $\Lambda^*$ if the following conditions are satisfied (refer again to \cref{fig:traversal} for illustrations): 
\begin{enumerate}[leftmargin=*,label={(C\arabic*)}]
	\item \label{c1:usedchannelsappear} For every channel $c$, there exists an $i\in [m]$ such that $c \in \mathcal{U}_i$ if and only if $c$ is used.
	\item \label{c2}  There exists a compatible edge ordering $\prec^{\mathrm{e}}$ for the restriction of $\Lambda^*$ to its nonhorizontal edges (recall that some visibility edges are horizontal) such that: 
		\begin{enumerate}
			\item 	\label{C2a:real_before}
				Let $e, e'$ be two edges that are incident to leaves and where $e \prec^{\mathrm{e}} e'$, let $c$ be the channel used 
				by $e$, and let $c'$ be the channel used by $e'$.
				Then there exist indices $i, i'$ such that $i \leq i'$ and $c \in \mathcal{U}_{i}$, $c' \in \mathcal{U}_{i'}$.
			\item \label{C2b:exclusivenes}
			Let $c_1,c_2\in U_{\mathrm{used}}$ such that for every edge $e_1$ using $c_1$ and for every edge $e_2$ using~$c_2$ we have $e_1 \prec^{\mathrm{e}} e_2$. Then there is no index $i\in [m]$ such that $\mathcal U_i$ contains both
				$c_1$ and $c_2$ (and every index for which $c_1$ is active is smaller than every index where $c_2$ is active).
			\item  \label{C2c:real_after}
				For every pair of used channels $c_1$, $c_2$ such that $c_2$ is being used by an edge $e$ that succeeds all edges that use $c_1$ in $\prec^{\mathrm{e}}$ there exists
				 an index $i$ such that $c_2 \in \mathcal{U}_i$ and $c_1 \notin \mathcal{U}_i$ (and $c_1$ is active for some index smaller than $i$).
		\end{enumerate}
\end{enumerate}

\begin{figure}[tbh]
\centering
    \includegraphics[page=3]{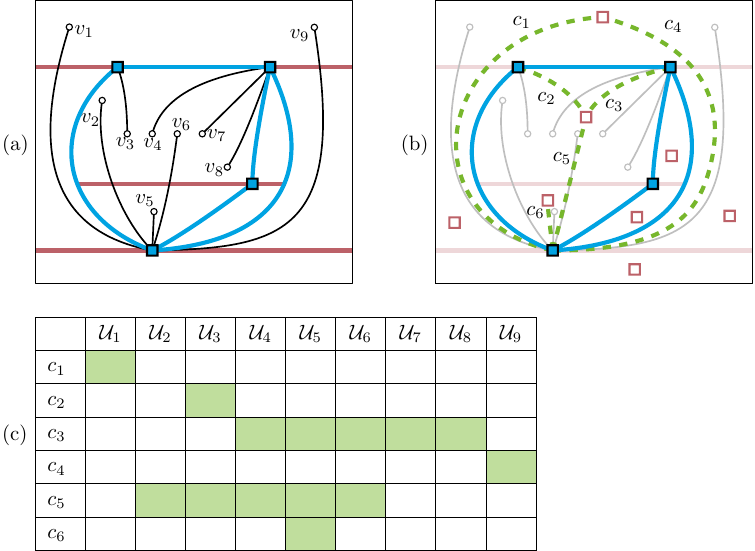}
  \caption{
  Like in all figures, filled square vertices belong to a vertex cover $C$ of the depicted graph.
  (a) The drawing $\Lambda^*$ together with its cell decomposition.
    (b)
    The dashed (green) ``edges'' represent the used channels $c_1 \dots, c_6$ of $\Lambda^*$ leading from vertices of $C$ to cells in $R$, which are represented by unfilled (red) squares.
    (c)
    A traversal sequence compatible with $\Lambda^*$ for which Property~\ref{c2} is satisfied for any compatible edge ordering that contains the edges of $v_1,v_2,\dots,v_9$ in this order.
    }
  \label{fig:traversal}
\end{figure}

\begin{restatable}[\restateref{cl:compat-traversal-exists}]{claim}{TraversalSequence}
\label{cl:compat-traversal-exists}
	There exists a traversal sequence $\mathcal U=(\mathcal{U}_1, \mathcal{U}_2, \dots, \mathcal{U}_m)$ that is compatible 
	with~$\Lambda^*$ and whose length is
	$m \in \mathcal O(k^2)$.
	Moreover,
	there is an algorithm that, given $\mathcal G$, $C$ and~$\Lambda^\mathrm{t}_{\mathrm {core}}$, computes a set of $2^{\mathcal O(k^2\log k)}$ traversal sequences that contains $\mathcal U$ in $2^{\mathcal O(k^2\log k)}n^{\mathcal O(1)}$ time.
\end{restatable}

\begin{restatable}[\restateref{claim:channels-unique-in-each-set}]{claim}{ChannelsUniqueInEachSet}
	\label{claim:channels-unique-in-each-set}
	Let $\mathcal U=(\mathcal{U}_1, \mathcal{U}_2, \dots, \mathcal{U}_m)$ be a traversal sequence that is compatible 
	with~$\Lambda^*$, let~$i\in [m]$, and let~$v$ be a leaf.
	Then $\mathcal U_i$ contains at most channel that can be used by~$v$.
\end{restatable}

\subparagraph{Ear orientations.}
Let $i\in [h]$ and let $V^\mathrm{e}_i\subseteq V_i$ be all non-crucial ears on the $i$th level.
Further, consider a mapping $s\colon V^\mathrm{e}_i \to \{\mathrm{left}, \mathrm{right}\}$.
We say that $s$ is an \emph{ear orientation} of level $i$.
We say that $s$ is \emph{valid} if it is possible to insert the ears~$V^\mathrm{e}_i$ (on the line $L_i$) along with their incident edges into $\Lambda^\mathrm{t}_{\mathrm {core}}$ such that
the resulting drawing~$\Lambda$ is crossing-free,
no constraint is violated (i.e., if $u\prec_i v$, then $u$ is placed to the left of~$v$),
and for every $v\in V^\mathrm{e}_i$, we have that $v$ is a left ear if and only if $s(v)=\mathrm{left}$.
We say that~$\Lambda$ is \emph{induced} by~$s$.
Note that  for any ear $v\in V^\mathrm{e}_i$ there is at most one left ear-face and at most one right ear-face into which it can be inserted without introducing crossings.
Hence, a valid ear orientation uniquely describes the ear-face in which each ear is placed.
Further, note that no two ears of $V^\mathrm{e}_i$ can be placed in the same ear-face without introducing crossings.
In contrast, whenever an ear orientation assigns only one ear to a given ear-face, it is possible to place the ear without introducing crossings.
These properties make it easy to test whether a given ear orientation is valid and, if so, construct the (unique) induced drawing in polynomial time.
In view of \cref{cl:2k-ears-per-level}, this means we can enumerate all valid ear orientations of a given level in $2^{2k}n^{\mathcal O(1)}$ time.

\subparagraph{Insertion sequences.}
Let $\mathcal U=(\mathcal{U}_1, \mathcal{U}_2, \dots, \mathcal{U}_m)$ be a traversal sequence that is compatible 
	with~$\Lambda^*$.
Further, let $i\in [h]$, let $s$ be a valid ear orientation of level $i$, and let~$\Lambda$ be its induced drawing. 
Finally, let $\mathcal{Q} =(Q_1, Q_2, \dots, Q_q)$ be a sequence with $Q_t = (v, j)$, $v \in V_i \cap V_{=1}$, and $1 \leq j \leq m$ for all $1 \leq t \leq q$.
We say $\mathcal{Q}$ is an \emph{insertion sequence} for $i$, $s$, and  $\mathcal U$ if the following conditions are fulfilled (for an example, see \cref{fig:insertion}):
\begin{enumerate}[leftmargin=*,label={(I\arabic*)}]
	\item \label{i1} 
		Let $v \in V_i \cap V_{=1}$. Then there exists at most one index $t$ such that $v \in Q_t$.
	\item \label{i2} 
		Let $Q_x  = (v_x, j_x)\in \mathcal{Q}$ and $Q_y = (v_y, j_y) \in \mathcal{Q}$ with $x < y$. Then $j_x \leq j_y$.
	\item \label{i3} 
		Let $Q_x= (v, j) \in \mathcal{Q}$. Then there exists a channel $c \in \mathcal{U}_j$ that can be used by $v$.
	\item \label{i4} 
		Let $Q_x= (v, j)\in \mathcal{Q}$. Then for every $w \in V_i \cap V_{=1}$ with $w \prec_i v$, there exists an 
		index $x' < x$ such that $w \in Q_{x'}$.
	\item \label{i5} 
		Let $Q_x= (v, j)\in \mathcal{Q}$ and let $(v,r,R) \in \mathcal{U}_j$ be the (unique, by \cref{claim:channels-unique-in-each-set}) channel usable by $v$ in $\mathcal{U}_j$. Then for every 
		$w \in V_i \setminus V_{=1}$ that is not a transition vertex in $r$ and where $v \prec_i w$, $w$ is 
		to the right of $r$ or on the right boundary of $r$.
		 Symmetrically, for every 
		$w \in V_i \setminus V_{=1}$ that is not a transition vertex in $r$ and where $w \prec_i v$, $w$ is
		to the left of $r$ or on the left boundary of $r$.
\end{enumerate}

\noindent Let $\mathcal{Q} = (Q_1, Q_2, \dots, Q_q)$ be an insertion sequence for $i$, $s$, and $\mathcal U$.
We say a leaf $v\in V_i \cap V_{=1}$ is \emph{choosable} with regard to $\mathcal{Q}$ if (i) there exists an index $j$, such that $(Q_1, Q_2, \dots, Q_q, (v, j))$ is an 
insertion sequence for $i$, $s$, and $\mathcal U$ as well and (ii) there exists no pair $v', j'$,  with $v' \in V_i \cap V_{=1}$ and $j' < j$ such that 
$\mathcal{Q} = (Q_1, Q_2, \dots, Q_q, (v', j'))$ is an insertion sequence. 

\begin{figure}[tbh!]
    \centering
    \includegraphics[width = \textwidth, page=3]{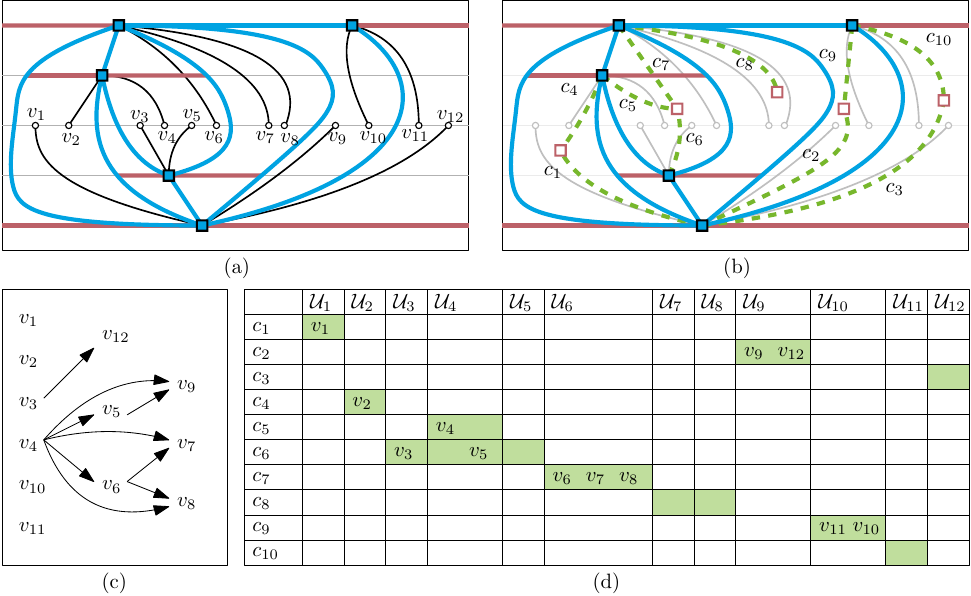}
  \caption{
  Like in all figures, filled square vertices belong to a vertex cover $C$ of the depicted graph.
  (a) The drawing $\Lambda^*$ together with its cell decomposition.
    (b)
    The dashed (green) ``edges'' represent the used channels $c_1, \dots, c_{10}$ of $\Lambda^*$ leading from vertices of $C$ to cells in $R$ represented by unfilled (red) squares.
Figure (c) shows the constraints between the vertices of the third level depicted in (a) and Figure~(d) illustrates the insertion sequences $\mathcal Q=((v_1, 1), (v_2, 2), (v_3, 3), (v_4, 4), (v_5, 4), (v_6, 6), (v_7, 6), (v_8, 6), (v_9, 9), (v_{12}, 9), (v_{11}, 10), (v_{10}, 10))$ for this level and the depicted traversal sequence $\mathcal U=(\mathcal U_1, \dots, \mathcal U_{12})$.
Note that the order of the vertices in $\mathcal Q$ is different from the one in $\Lambda^*$.
}
  \label{fig:insertion}
\end{figure}

\begin{restatable}[\restateref{cl:insertion-sequence-exists}]{claim}{InsertionSequenceExists}
\label{cl:insertion-sequence-exists}
Let $\mathcal U=(\mathcal{U}_1, \mathcal{U}_2, \dots, \mathcal{U}_m)$ be a traversal sequence that is compatible 
	with~$\Lambda^*$ and let $\prec^{\mathrm{e}}$ be a compatible edge ordering for the restriction of $\Lambda^*$ to its nonhorizontal edges, for which Property~\ref{c2} is fulfilled for $\mathcal{U}$.
Further, let $i\in [h]$ and let $s$ be the valid ear orientation of level~$i$ that is used in~$\Lambda^*$.
For every $q \in \{0, 1, \dots, |V_i \cap V_{=1}|\}$, there exists an insertion sequence 
$\mathcal{Q}_q = (Q_1, Q_2, \dots, Q_q)$ for $i$, $s$, and $\mathcal U$
such that the following two properties are satisfied (for an example, see \cref{fig:insertion}):

\noindent \emph{Interval property:}
For every vertex $v \in V_i \cap V_{=1}$
there exists at most one nonempty maximal interval $[a, b]$ where $0\le a\le b\le q$ such that $v$ is choosable with regard to $\mathcal Q_j$ if 
and only if $a \leq j \leq b$. If such an interval $[a, b]$ exists, then $b=q$ or $v\in Q_{b+1}$.
Conversely, if $v$ occurs in some entry of $\mathcal Q_q$, then the interval exists. 

\noindent \emph{Dominance property:} Let $\prec^l_i$ be the restriction of $\prec^{\mathrm{e}}$ to edges incident to leaves on level~$i$. Further, let $Q_\ell = (v_\ell, j_\ell)\in \mathcal Q_q$, let $v_{\ell'} \in V_i \cup V_{=1}$, and let $e_\ell$ ($e_{\ell'}$) be the edge incident to $v_\ell$ ($v_{\ell'}$). Then, if $e_\ell \prec^l_i e_{\ell'}$ or $v_\ell  = v_{\ell'}$, we have $j_\ell \le j'$, where $j'$ is the maximum index such that the channel $c'$ used by $v_{\ell'}$ (in $\Lambda^*$) is in $\mathcal U_{j'}$.

\noindent Moreover, $\mathcal{Q}_k$ is a prefix of $\mathcal{Q}_{k+1}$ for  all $0 \leq k \leq |V_i \cap V_{=1}|-1$.
Finally, there is an algorithm that,
given $\mathcal G$, $C$, $\Lambda^\mathrm{t}_{\mathrm {core}}$,  $\mathcal U$, $i$ and~$s$, 
computes 
$\mathcal{Q}_{|V_i \cap V_{=1}|}$ in polynomial time.
\end{restatable}

\subparagraph{Computing the drawing.}
When given a traversal sequence that is compatible with~$\Lambda^*$,
we can utilize \cref{cl:insertion-sequence-exists} to obtain a valid ear orientation together with an insertion sequence for a given level by simply trying to apply the algorithm corresponding to \cref{cl:insertion-sequence-exists} for all valid ear orientations of that level.
We do so for each level and then use the gathered information to construct the desired drawing by means of the following claim.
We remark that when the algorithm corresponding to \cref{cl:insertion-sequence-exists} successfully terminates, it is guaranteed to return an insertion sequence for the given valid ear orientation.
It might output an insertion sequence even if the given valid ear orientation is not the one used in~$\Lambda^*$.
However, this does not invalidate our strategy as the following claim does \emph{not} require that the given ear orientations are the ones used in~$\Lambda^*$.

\begin{restatable}[\restateref{cl:generate-drawing-from-ear-orientation}]{claim}{GenerateDrawingFromEarOrientation}
\label{cl:generate-drawing-from-ear-orientation}
There is an algorithm that, given $\mathcal G$, $C$, $\Lambda_\mathrm{core}$, $\Lambda^\mathrm{t}_{\mathrm {core}}$, 
a traversal sequence $\mathcal U=(\mathcal{U}_1, \mathcal{U}_2, \dots, \mathcal{U}_m)$ that is compatible 
	with~$\Lambda^*$,
 and, for each level $i\in [h]$, a valid ear orientation $s^i$, as well as an insertion sequence $\mathcal{Q}^i =(Q^i_1, Q^i_2, \dots, Q^i_{q^i})$
for $i$, $s^i$, and $\mathcal U$
such that $q^i = |V_i \cap V_{=1}|$ (that is, $\mathcal{Q}^i$ contains all leaves of level $i$),
computes an extension of $\Lambda_\mathrm{core}$ whose restriction to $G$ is a constrained level planar drawing 
of~$\mathcal {G}$ in polynomial time.
\end{restatable}

\subparagraph{Wrap-up.}
In the beginning of (and throughout the) proof of \cref{lem:insertLeftovers},
we have already sketched how the individual pieces of the proof fit together.
We formally summarize our strategy in the proof of the following claim.

\begin{restatable}[\restateref{cl:summaryStep3}]{claim}{SumStepThree}
\label{cl:summaryStep3}
There is an algorithm that, given~$\mathcal G$, $C$, $\Lambda_\mathrm{core}$ and~$\Lambda^\mathrm{t}_{\mathrm {core}}$, extends $\Lambda_{\mathrm {core}}$ to a drawing $\Gamma$ whose restriction to $G$ is a constrained level planar drawing of~$\mathcal G$ in $2^{\mathcal O(k^2\log k)}n^{\mathcal O(1)}$ time.
\end{restatable}

\noindent This concludes the proof of \cref{lem:insertLeftovers}.
\end{proof}

\subparagraph{Summary.}
In the beginning of \cref{sec:algorithm}, we have already sketched how \cref{lem:isolated,lem:step1,lem:insertTransition,lem:insertLeftovers} can be combined to obtain 
\cref{thm:fpt}.
We formally summarize:

\begin{restatable}[\restateref{thm:main-result}]{theorem}{Main}
\label{thm:main-result}
There is an algorithm that, given a constrained level graph $\mathcal G=(G=(V,E),\gamma,(\prec_i)_i)$ and a vertex cover~$C$ of~$G$, either constructs 
a constrained level planar drawing of~$\mathcal G$ or correctly 
 reports that such a drawing does not exist in time $2^{\mathcal O(k^2\log k)}\cdot n^{\mathcal O(1)}$, where $n=|V|$ and $k=|C|$.
\end{restatable}

Given that a smallest vertex cover can be obtained in FPT-time with respect to its size~\cite{DBLP:books/sp/CyganFKLMPPS15}, our main result (\cref{thm:fpt}) follows from \cref{thm:main-result}.

\section{Discussion}
\label{sec:conclusion}

We have shown that \CLP is FPT when parameterized by the vertex cover number.
A speed-up to polynomial time or a generalization to the smaller graph parameters (in particular, tree-depth, path-width, tree-width, and feedback vertex set number) is not possible, even if restricting to OLP or proper CLP instances.

Recall from \cref{sec:intro} that in the \textsc{Level Planarity} variant \PLPlong (\PLP),
a given level planar drawing of a subgraph of the input graph~$\mathcal G$ has to be extended to a
full 
drawing of~$\mathcal G$, 
which can be seen as a generalization of \OLP and, in the proper case, a specialization of \CLP.
An instance of \PLP can always be turned into an equivalent proper instance (and, thus, a \CLP instance) by subdividing each 
edge on each level it passes through.
However, in general this operation will (dramatically) increase the vertex cover number of the instance.
Hence, our techniques cannot (directly) be applied.
It thus is an interesting problem to study the parameterized complexity of \PLP with respect to the vertex cover number.



\bibliography{olp-bib}

\newpage

\appendix

\section{Material omitted from \cref{sec:types} (Preliminaries)}

\Isolated*
\label{lem:isolated*}

\begin{proof}
For each level~$V_\ell$, the drawing~$\Gamma'$ induces a linear order~$<_\ell'$ on $V_\ell'$ that extends $\prec'_\ell$.
Consider the order $<_\ell=(<_\ell'\cup \prec_\ell)$.
We claim that~$<_\ell$ is a partial order, implying that the desired drawing can be obtained by performing a topological sort on each level.

To prove the claim, assume towards a contradiction that~$<_\ell$ contains a cycle~$\delta$ of the form $v_1 <_\ell v_2 <_\ell \dots <_\ell v_i < v_1$.
Without loss of generality, we may assume that $<_\ell$ contains no cycle involving fewer vertices.
Since $<_\ell'$ and $\prec_\ell$ are partial orders, the cycle~$\delta$ contains at least one isolated vertex and at least one non-isolated vertex.
Without loss of generality, we may assume that~$v_1$ is isolated and~$v_i$ is non-isolated.
Let~$j$ ($<i$) be the largest index such that $v_1,v_2,\dots,v_j$ are all isolated.
Since~$\prec_\ell'$ only involves non-isolated vertices, it follows that $v_i \prec_\ell v_1 \prec_\ell v_2 \prec_\ell \dots \prec_\ell v_{j+1}$ and, further, 
that $j+1<i$ since~$\prec_\ell$ is a partial order.
However, this implies that $v_{j+1}<_\ell v_{j+2} <_\ell \dots <_\ell v_i <_\ell v_{j+1}$ is a cycle in~$<_\ell$ involving fewer vertices than~$\delta$; a contradiction.
So indeed, $<_\ell$ is a partial order, as claimed.
\end{proof}

\TwoK*
\label{cl:2k-ears-per-level*}

\begin{proof}
Let $i\in [h]$.
We will use a charging argument to show that $V_i$ contains at most $2|C|$ ears.
To this end, let $\Gamma'$ denote the subdrawing of~$\Gamma$ that contains only the ears in $V_i$, as well as all vertices in~$C$.
Let~$v$ be a vertex in~$C$ that is closest to the line~$L_i$.
By planarity, this vertex can be adjacent to at most two ears of~$V_i$.
We charge these ears to~$v$ and then delete them, as well as~$v$ from~$\Gamma'$.
We repeat this strategy until all vertices in~$C$ (and, thus, all ears in $V_i$) have been deleted.
Every ear is charged to exactly one vertex in~$C$ and every vertex in~$C$ is charged at most twice, which implies the claim.
\end{proof}

\section{Material omitted from \cref{sec:vis-core} (Visibility extensions and cores)}

\Cconnected*
\label{lem:VC-connected*}

\begin{proof}
We begin by proving the connectivity of~$\Lambda_C$.
Let~$\Lambda'$ denote the (level planar) subdrawing of~$\Lambda$ that is induced by the non-horizontal edges.
Let $S=(e_1,e_2,\dots,e_{m'})$ be a sequence of edges corresponding to a compatible edge ordering of~$\Lambda'$.
We will show inductively that for each $1\le i\le m'$ the vertices of the subset~$V_{E_i}\subseteq C$ of vertices in~$C$ that are incident to at least one 
of the edges in $E_i=\{e_1,e_2,\dots,e_i\}$ belong to the same connected component of~$\Lambda_C$,
which establishes the connectivity of~$\Lambda_C$ since $V_{E_{m'}}=C$.

Clearly the claim holds for $i=1$.
Now let $1\le j<m'$ and assume that the claim holds for each $i\le j$.
We will show that the claim also holds for $i=j+1$.
Clearly, this is the case when $e_{j+1}$ is incident to a vertex of $V_{E_i}$, so assume otherwise.
Let~$v$ be a vertex of~$e_{j+1}$ that belongs to~$C$.
Consider the horizontal ray~$r_v$ emanating from~$v$ to the left.
To prove the claim, it suffices to show that $v$ belongs to the connected component of~$\Lambda_C$ that contains~$V_{E_i}$.

If~$r_v$ intersects~$\Lambda'$ only in~$v$, then~$\Lambda$ contains a visibility edge that joins~$v$ with a vertex of~$C$ incident to~$e_1$ (by 
definition of $S$, $e_1$ is visible from the left in~$\Lambda$).
In this case, the claim follows since this vertex belongs to~$V_{E_i}$.

So assume that~$r_v$ intersects~$\Lambda'$ not only in~$v$.
Let~$p$ denote the first point $\neq v$ of~$\Lambda'$ that is encountered when traversing~$r_v$ from~$v$.
Since there are no isolated vertices, $p$ is located on some edge~$e_t$.
Moreover, by definition of~$S$, $e_{j+1}$, and $v$, it follows that~$t<j+1$.
The definition of~$p$ asserts that there is a visibility edge between $v$ and a vertex of~$C$ incident to~$e_t$.
The claim follows since this vertex belongs to~$V_{E_i}$, which concludes the proof of the connectivity of~$\Lambda_C$.

To bound the number of edges, assume that~$\Lambda_C$ contains at least three edges (otherwise, the statement holds).
The drawing~$\Lambda_C$ may contain parallel edges, however, the definition of~$\Lambda$ and the connectivity of~$\Lambda_C$
assert that each 
face of~$\Lambda_{\mathrm{core}}$ (except for possibly the outer face, which could be bounded by two parallel edges) has at least three edges on its boundary.
Consequently, Euler's polyhedron formula yields that $\Lambda_{\mathrm{core}}$ has at most $3k-6+1\in \mathcal O(k)$ edges (the $+1$ accounts for the special role of the outer face).
\end{proof}

\SizeCore*
\label{lem:size-of-core*}

\begin{proof}
We will prove the statement for non-refined visbility extensions; the statement for refined visibility extensions then follows easily by definition.
Recall that the core of~$\Gamma$ with respect to~$C$ is the subset of~$V$ that contains $C$, $V_{\ge 3}(C)$, as well as all crucial vertices of~$\Gamma$ with respect to~$C$.

\subparagraph{Connectivity.}
We begin by showing that~$\Lambda_{\mathrm{core}}$ is connected.
By \cref{lem:VC-connected}, the subdrawing~$\Lambda_C$ of~$\Lambda$ (and $\Lambda_{\mathrm{core}}$) that is induced by~$C$ is connected.
Moreover, since~$C$ is a vertex cover, every vertex of $V_{\ge 3}(C)$, as well as every crucial vertex is adjacent to some vertex of~$C$.
Consequently, $\Lambda_{\mathrm{core}}$ is connected, as claimed.

\subparagraph{Counting crucial vertices.}
Recall that the crucial vertices~$\Gamma$ with respect to~$C$ are the ears that are outermost, innermost, or bounding ears of some pair of vertices of~$C$.
We will examine the counts of these different types of vertices individually.

\subparagraph{Counting outer/innermost ears.}
We begin by bounding the number of vertices that are outermost left ears.
Let $V_{\mathrm{outer}}^{\mathrm{left}}\subseteq V_{=2}(C)$ be the set of all these ears.
Let~$C'\subset C$ be the set of vertices that are adjacent to at least one ear in  $V_{\mathrm{outer}}^{\mathrm{left}}$.
We will now construct a graph drawing~$\Pi$ as follows:
the vertex set of~$\Pi$ is~$C'$ and its vertices are placed exactly as in~$\Lambda$.
Let $v\in V_{\mathrm{outer}}^{\mathrm{left}}$ and let $c_a,c_b\in C'$ denote the neighbors of~$v$.
We add an edge between $c_a,c_b$ in~$\Pi$ and use the simple curve formed by $v$ and its incident edges in~$\Lambda$ to draw it.
This process is repeated for every ear $v\in V_{\mathrm{outer}}^{\mathrm{left}}$.
The resulting drawing~$\Pi$ has $|C'|\le k$ vertices and $|V_{\mathrm{outer}}^{\mathrm{left}}|$ edges, which are in one-to-one correspondence 
with the ears in $V_{\mathrm{outer}}^{\mathrm{left}}$.
Moreover, the drawing is by construction planar and has no parallel edges.
Hence, Euler's polyhedron formula implies $|V_{\mathrm{outer}}^{\mathrm{left}}|\le 3k-6$.

The outermost right and the innermost left/right ears can be bounded analogously, which yields $\mathcal O(k)$ outermost and innermost ears in total.

\subparagraph{Counting bounding ears.}
To count the number of bounding ears, we will use a charging argument.
Consider a matching pair~$v_1,v_2$ of bounding ears of~$c_a,c_b\in C$ with region~$r$ (which, by definition, contains at least one vertex of~$C$).
If there is no pair of matching bounding ears whose region is a proper subset of~$r$, we \emph{charge} the pair~$v_1,v_2$ to an arbitrary vertex of~$C$ in the interior of~$r$.
Otherwise, let~$r'\subset r$ be an inclusion-maximal region contained in~$r$ that is the region of a pair~$v_1',v_2'$ of matching bounding ears.
The boundary of~$r'$ contains exactly four vertices: $v_1'$ and $v_2'$, as well as two vertices of~$C$.
At least one of the latter two has to be contained in the interior of~$r$ since $r'\subset r$, say~$c\in C$.
Note that~$c$ is distinct from both~$c_a$ and~$c_b$, which belong to the boundary of~$r$.
We \emph{charge} the pair~$v_1,v_2$ to~$c$.
This charging procedure is repeated for all matching pairs of bounding ears.

The inclusion-maximality of the regions~$r'$ implies that each vertex of~$C$ gets charged at most once:
consider a vertex~$c$ and assume it is charged by a pair~$v_1,v_2$ of bounding ears with region~$r$.
Then~$c$ is located in the interior of~$r$.
Suppose there is another pair~$v_1',v_2'$ of bounding ears with region~$r'$ that is charged to~$c$.
Then~$c$ is also located in the interior of~$r'$.
By planarity and the fact that~$c$ is interior to both~$r$ and~$r'$, the boundaries of~$r$ and~$r'$ are nested, leaving two cases:
$r'\subset r$ or $r\subset r'$.
In the case $r'\subset r$, we obtain a contradiction since~$v_1,v_2$ would have been charged to a vertex~$\tilde c$ that is located on the boundary 
of~$r'$ or on the boundary of some region~$r''$ with~$r'\subset r''\subset r$, implying that~$\tilde c$ is distinct from~$c$, which is interior to~$r'$.
The case $r'\subset r$ is symmetric.
So indeed, each vertex of~$C$ is charged at most once, as claimed.
Moreover, no vertex of the outer (i.e., unbounded) face of~$\Gamma$ can be charged.
Since there are at least two vertices of $C$ on the outer face (assuming $|E|\ge 3$), we obtain that the number of pairs of matching bounding ears is at most $k-2$.
Consequently, the number of bounding ears is at most $2k-4\in \mathcal O(k)$.

\subparagraph{Wrap-up.}
We have shown that the number of crucial vertices is~$\mathcal O(k)$.
The number of vertices in~$C$ is $k$ by definition and the number of vertices in $V_{\ge 3}(C)$ is~$\mathcal O(k)$ by \cref{lem:xg2:Vgeq3InKexp3}.
This yields a total number of $n'\in \mathcal O(k)$ vertices in the core of~$\Gamma$ with respect to~$C$.

The $\mathcal O(k)$-bound on the number of visibility edges follows from \cref{lem:VC-connected}.
When removing all visibility edges from $\Lambda_{\mathrm{core}}$, we end up with a drawing without parallel edges.
Hence, Euler's polyhedron formula implies that its number of edges, and therefore the number of edges in $\Lambda_{\mathrm{core}}$, is bounded by $3n'-6\in \mathcal O(k)$.
\end{proof}

\section{Omitted material from \cref{sec:algorithm} (Algorithm)}

\StepOne*
\label{lem:step1*}

\begin{proof}
We proceed in two steps.
In the first step, we show that there exists a refined visibility extension~$\Lambda^*$ of~$\Gamma^*$ such that the set of vertices in~$\Lambda^*$ that do not belong to~$\mathcal G$ only use a restricted (small) set of levels.
In the second step, we use the results from the first step to discuss how to construct the desired family~$\mathcal F$.

\subparagraph{Existence of $\Lambda^*$.}
Towards the first step, let~$\Lambda$ be a visibility extension of~$\Gamma^*$.
We apply the following normalization:
let $a,b$ be a pair of matching bounding ears and let $c,d\in C$ be the common neighbors of~$a,b$.
Without loss of generality, assume that the level of $b$ is lower than the one of $a$ and that~$a,b$ are top ears.
If the intersection~$I$ of the region~$R$ of $a,b$ with the closed upper half-plane bounded by $L_{\gamma(b)}$ is empty or contains only leaves of~$a,b$,
then we redraw all edges of leaves of $c$ whose upper endpoint is in~$I$ such that they intersect $L_{\gamma(b)}$ directly to the right of $(c,a)$ and then closely follow $(c,a)$ towards~$c$ such that none of them has a vertex in $R\setminus I$ to its left.
Similarly, we can now redraw $(c,b)$ such that it also has no vertex in $R\setminus I$ to its left.
We apply a symmetric transformation to the leaves of~$d$ and the edge~$(d,b)$.
Finally, we insert a maximal set of visibility edges to ensure that the result is again a visibility extension of~$\Gamma^*$.
We repeat this strategy until the precondition of the redrawing procedure is unfulfilled for every pair of matching bounding ears.
A given edge might be redrawn multiple times, however, this procedure is guaranteed to terminate since each edge either moves monotonically to the left or moves monotonically to the right.
Overloading our notation, from now on we use~$\Lambda$ to denote the final redrawing.

We will now describe an iterative process that augments~$\Lambda$ (by inserting new ears) to obtain the desired refined visbility extension~$\Lambda^*$.
Each new vertex created in this process will be placed on a ``private'' new level (of width $1$).
To facilitate the description of the extended level assignment, we introduce the following terminology:
let $x$ be a vertex in a level planar drawing (possibly augmented by some horizontal edges).
Let $\ell$ be the y-coordinate of~$x$ and let~$\ell'$ be the largest y-coordinate of a vertex below $x$ (if there no such vertex, we set~$\ell'=\ell-1$).
We say that the line $L_{(\ell+\ell')/2}$ is \emph{directly below} $x$.
The line \emph{directly above} $x$ is defined symmetrically.

Suppose we have already performed some number (possibly zero) of augmentation steps and let~$\Lambda'$ denote the resulting augmentation of~$\Lambda$.
We now discuss how to perform the next augmentation step.
Let $v\in V$ be an ear of~$\mathcal G$ that is crucial in~$\Lambda'$ and let $u,w\in C$ denote its neighbors (if such a vertex~$v$ does not exist, we are done and~$\Lambda'$ is the desired representation~$\Lambda^*$).
We describe a procedure that creates up to two new ears to ensure that~$v$ is no crucial ear of the resulting augmentation of~$\Lambda'$.
Without loss of generality, let~$v$ be a top left ear (the other cases are handled symmetrically).
We iteratively consider four cases, which are not mutually exclusive (up to two can occur at the same time).

\begin{figure}[ptbh!]
    \centering
    \includegraphics[width=\textwidth]{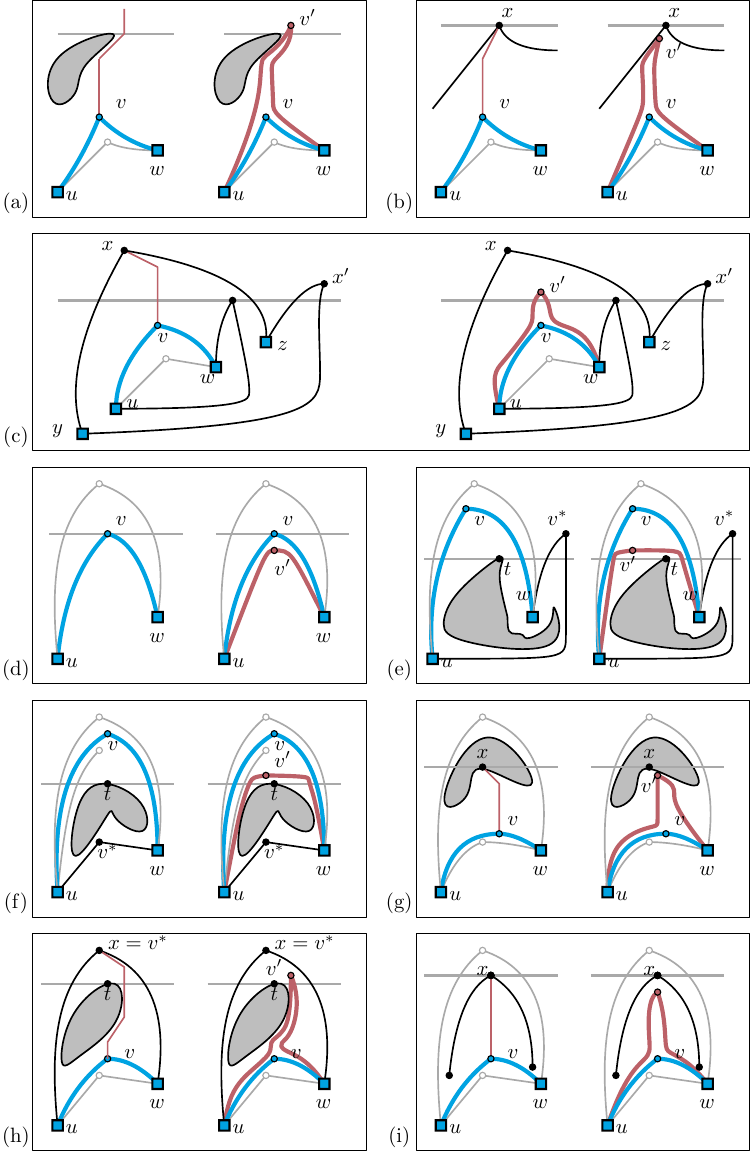}
  \caption{
  The set of operations and cases used in the construction of~$\Lambda^*$.
  Like in all figures, filled square vertices belong to a vertex cover $C$ of the depicted graph and
    filled (round or square) vertices belong to the core of the shown drawing with respect to~$C$.
  }
  \label{fig:crucials}
\end{figure}

\subparagraph{Case 1.} $v$ is an outermost ear.
Let $\phi$ be a y-monotone curve whose lower endpoint is $v$,
whose interior does not intersect~$\Lambda'$, and
whose upper endpoint is either (i) on the line directly above the top-most vertex of $\Lambda'$ 
or (ii) a vertex~$x$ of~$\Lambda'$ so that~$\phi$ is between two incoming edges of~$x$.
(Such a curve does always exist: it can be constructed by traversing upwards from $v$ in a greedy-like fashion.)

If the upper endpoint of $\phi$ is of type (i), we place a new ear $v'$ of $u,w$ at the upper endpoint of~$\phi$ and draw its incident edges by first following~$\phi$ to~$v$ and then either the edge $(u,v)$ or $(w,v)$; see \cref{fig:crucials}(a). 

If the upper endpoint of~$\phi$ is of type (ii), then it is a vertex~$x$ such that either $x \in C$, $x \in V_{\ge 3}(C)$, or $x$ is a top ear of~$\Lambda'$. We distinguish two subcases.

	If $x \in C$, $x \in V_{\ge 3}(C)$, or $x$ is the lowest top ear of its two neighbors,
	we place a new ear $v'$ of $u,w$ on the intersection of~$\phi$ with the line directly below~$x$ and draw its incident edges by first following~$\phi$ to~$v$ and then either the edge $(u,v)$ or $(w,v)$; see \cref{fig:crucials}(b).
	
	If $x$ is not the lowest top ear of its two neighbors, which we denote by $y$ and $z$, we place a new ear $v'$ of $u,w$ on the intersection of~$\phi$ with the line~$L$ directly above the top-most ear of $u,w$ and draw its incident edges by first following~$\phi$ to~$v$ and then either the edge $(u,v)$ or $(w,v)$; see \cref{fig:crucials}(c).
	Note that~$\phi$ indeed always intersects~$L$:
	all ears of $u,w$ lie inside the region~$R$ bounded by the cycle $(x,z,x',y)$ where $x'$ is the lowest top ear of $y,z$. It follows that the top-most ear of $u,w$ (and, thus, $L$) is below~$x$ since the latter is the topmost point of~$R$.
	
	Overloading our notation, from now on (when considering Cases 2--4) we use~$\Lambda'$ to denote the resulting drawing in which~$v$ no longer satisfies the assumption of Case~1.

\subparagraph{Case 2.}
	$v$ is an innermost ear.
	If $v$ is a lowest top ear of $u,w$, 
	we place a new vertex~$v'$ on the line directly below~$v$ and between the two edges incident to~$v$ and draw the two incident edges of~$v'$ by following the incident edges of~$v$; see \cref{fig:crucials}(d).
	Otherwise, there exists an innermost right ear $v^*$ of $u,w$ whose level is lower than the level of~$v$ (since~$v$ is an innermost left ear); for illustrations refer to \cref{fig:crucials}(e).
	 Among the vertices that are located in the region bounded by the cycle $(u, v^*, w,v)$ and that are not leafs of $u$ or $w$, let $t$ be a top-most vertex (if the region does not contain such a vertex, we set $t = v^*$).
	 Assume without loss of generality that the edge~$(u,v)$ is to the left of the edge~$(w,v)$.
	 We place a new vertex~$v'$ on the line directly above~$t$ and between the two edges incident to~$v$ and to the right of the rightmost edge that is incident to a leaf of~$u$ (if any) and to the left of the leftmost edge that is incident to a leaf of~$w$ (if any).
	 We connect $v'$ to $u$ and $w$ by going to the left and right until reaching edges. These edges must either be incident to $v$ or to a leaf of $u$ or $w$. In both cases, we can follow these edges to connect $v'$ to $u$ and $w$.
	 
	 Overloading our notation, from now on (when considering Cases 3--4) we use~$\Lambda'$ to denote the resulting drawing in which~$v$ no longer satisfies the assumptions of Cases~1--2.
	 
	 We conclude this case with a small observation that will later help with the construction of the desired family~$\mathcal F$:
	if $t \neq  v^*$, then $t$ can neither be a leaf of $u$ or $w$ nor an ear of $u$ and $w$.
	It follows that in the graph obtained by removing $u$ and $w$ from~$\Lambda'$ there is at least one vertex of $C$ in the connected component of~$t$.
	The number of connected components that contain a vertex of~$C$ is bounded by $\mathcal O(k)$.
	Hence, the number of levels with a top-most vertex of one of these components is also bounded by $\mathcal O(k)$.

\subparagraph{Case 3.}
There is a vertex $v^*$ such that $v,v^*$ is a pair of matching bounding ears of $u,w$ in $\Lambda'$ and $v^*$ has a lower level than $v$.
Assume without loss of generality that the edge~$(u,v)$ is to the left of the edge~$(w,v)$; for illustrations refer to \cref{fig:crucials}(f).
Among the vertices that are located in the region bounded by the cycle $(u, v^*, w,v)$ and that are not leaves of $u$ or $w$, let $t$ be a top-most vertex.
If the level of $v^*$ is above the one of~$t$, by our normalizing assumption, we know that~$v^*$ is not a vertex of~$\Lambda$.
In this case re-define $t = v^*$.
We place a new vertex~$v'$ on the line directly above~$t$ and between the two edges incident to~$v$ and to the right of the rightmost edge that is incident to a leaf of~$u$ (if any) and to the left of the leftmost edge that is incident to a leaf of~$w$ (if any).
We connect $v'$ to $u$ and $w$ by going to the left and right until reaching edges. These edges must either be incident to $v$ or to a leaf of $u$ or $w$. In both cases, we can follow these edges to connect $v'$ to $u$ and $w$.

Overloading our notation, from now on (when considering Case 4) we use~$\Lambda'$ to denote the resulting drawing in which~$v$ no longer satisfies the assumptions of Cases~1--3.

\subparagraph{Case 4.}
There is a vertex $v^*$ such that $v,v^*$ is a pair of matching bounding ears of $u,w$ in $\Lambda'$ and $v^*$ has a higher level than $v$.
Similiar to Case~1, let $\phi$ be a y-monotone curve whose lower endpoint is $v$,
whose interior does not intersect~$\Lambda'$, and
whose upper endpoint is a vertex~$x$ of~$\Lambda'$ so that~$\phi$ is between two incoming edges of~$x$.
(Such a curve does always exist: it can be constructed by traversing upwards from $v$ in a greedy-like fashion.)

If $x \in C$ or $x \in V_{\ge 3}(C)$,
	we place a new ear $v'$ of $u,w$ on the intersection of~$\phi$ with the line directly below~$x$ and draw its incident edges by first following~$\phi$ to~$v$ and then either the edge $(u,v)$ or $(w,v)$; see \cref{fig:crucials}(g).

If $x=v^*$, then among the vertices that are located in the region bounded by the cycle $(u, v^*, w,v)$ and that are not leaves of $u$ or $w$, let $t$ be a top-most vertex.
We place a new ear $v'$ of $u,w$ on the intersection of~$\phi$ with the line~$L$ directly above $t$ and draw its incident edges by first following~$\phi$ to~$v$ and then either the edge $(u,v)$ or $(w,v)$; see \cref{fig:crucials}(h).
	Note that~$\phi$ indeed always intersects~$L$ by our normalization procedure
	(this is true even if $v^*$ does not belong to $\Lambda$ by construction).

Finally, assume that $x$ is a top ear, but at least one of its adjacent vertices is neither $u$ nor $w$.
Let $y,z$ be the neighbors of $x$.
We place a new ear $v'$ of $u,w$ on the intersection of~$\phi$ with the line directly below~$x$ and draw its incident edges by first following~$\phi$ to~$v$ and then either the edge $(u,v)$ or $(w,v)$; see \cref{fig:crucials}(i).
	Note that $x$ is a lowest top ear of $y,z$ since~$\phi$ is crossing-free.

This concludes the description of Case~4, after which~$v$ no longer satisfies the assumptions of Cases~1--4.

We repeat this augmentation step until no crucial vertices of $\Lambda$ are still crucial in $\Lambda'$ yielding the desired refined visibility representation~$\Lambda^*$.
Note that,  by \cref{lem:size-of-core}, the total number of augmentation steps is~$\mathcal O(k)$.

\subparagraph{Construction of~$\mathcal F$.}
We are now ready to discuss the construction of the desired family~$\mathcal F$.
Let~$\Lambda_{\mathrm{core}}$ be the subdrawing of $\Lambda^*$ induced by the core of~$\Lambda^*$ with respect to~$C$.
The drawing $\Lambda_{\mathrm{core}}$ is uniquely described by $\mathcal G$, $C$, the set of visibility edges of $\Lambda_{\mathrm{core}}$ (and $\Lambda^*$), the set of crucial ears of $\Lambda_{\mathrm{core}}$ (and $\Lambda^*$) together with their levels and their incident edges, and a compatible edge ordering of the nonhorizontal edges of $\Lambda_{\mathrm{core}}$.
The graph~$\mathcal G$, as well as the vertex cover $C$ are given, so it suffices to enumerate all possible options for the remaining elements.

There are $m_{\mathrm{vis}}\in \mathcal O(k)$ visibility edges by \cref{lem:size-of-core}
and each of these visibility edges joins a pair of vertices in $C$.
Hence, there are at most $\binom{k}{2}^{m_{\mathrm{vis}}}\subseteq k^{\mathcal O(k)}\subseteq 2^{\mathcal O(k\log k)}$ possible options for choosing the set of visibility edges.

To enumerate the set of crucial ears along with their level assignment,
we mimic the above construction of~$\Lambda^*$:
we first enumerate all options to pick the pair of neighbors of the first new vertex along with its level,
then, for each of these options, we enumerate all options to pick the pair of neighbors of the the second vertex along with its level, etc.,
until we have obtained all options to pick $\mathcal O(k)$ vertices together with their levels.
More precisely, suppose we have already enumerated all options to pick the first~$i$ vertices together with their neighbors and levels.
For each of these options, to enumerate all options to pick the next vertex~$v'$, we go through all ways to pick its two neighbors~$u,w\in C$ and through all ways to pick the level of~$v'$.
There are $\mathcal O(k)$ pairs of vertices in~$C$ with ears by \cref{lem:size-of-core}.
To bound the number of ways to pick the level of~$v'$, recall that whenever the above procedure places a vertex~$v'$, it is assigned to a new level directly above or below a level of one of the following categories:
\begin{itemize}
\item a level with a vertex in $C$ ($\mathcal O(k)$ possibilities),
\item a level with a vertex in $V_{\ge 3}(C)$ ($\mathcal O(k)$ possibilities by \cref{lem:xg2:Vgeq3InKexp3}),
\item a level of a vertex that does not belong to~$\mathcal G$, i.e., a level used for
one of the already placed vertices ($\mathcal O(i)\subseteq \mathcal O(k)$ possibilities),
\item a level with a top-most or bottom-most vertex of $\mathcal G$ ($\mathcal O(1)$ possibilities),
\item a level with a top-most top ear, a top-most bottom ear, a bottom-most top ear, or a bottom-most bottom ear of some pair of vertices in~$C$ ($\mathcal O(k)$ possibilities by \cref{lem:size-of-core}),
\item a level with a top-most or bottom-most vertex of a connected component that contains a vertex of~$C$ in the graph obtained by removing $u$ and $w$ from the current graph ($\mathcal G$ together with the visbility edges and the already added vertices) ($\mathcal O(k)$ possibilities by the observation at the end of Case~2).
\end{itemize}
In total, for a fixed pair of neighbors~$u,w$, there are thus $\mathcal O(k)$ options to pick a level for~$v'$.
We immediately discard level assignments for which~$v'$ is no ear.
By multiplying with the number of ways to choose the neighbors, we obtain~$\mathcal O(k^2)$ options to choose~$v'$ and its level.
Multiplying the number of options for all $\mathcal O(k)$ steps together, we obtain a total number of $k^{\mathcal O(k)}\subseteq 2^{\mathcal O(k\log k)}$ ways to create the set of crucial ears along with their levels.
By multiplying with the number of ways to choose the visibility edges, we obtain a total of $2^{\mathcal O(k\log k)}$ options to choose the graph that corresponds to~$\Lambda_{\mathrm{core}}$.
For each of these options we enumerate all $k^{\mathcal O(k)}\subseteq 2^{\mathcal O(k\log k)}$ permutations of the set of non-horizontal edges and, interpreting the permutation as a compatible edge order, try to construct a level planar drawing for which this order is compatible (cf.\ \cref{sec:types}).
If we succeed,
we check whether the drawing is conform with $(\prec_i)_i$ and can be augmented with the horizontal visibility edges. If so,
we include the drawing in the set~$\mathcal F$ of reported drawings.
The size of the thereby constructed set~$\mathcal F$ is bounded by~$2^{\mathcal O(k\log k)}$ and it is guaranteed to contain $\Lambda_{\mathrm{core}}$ by construction.
\end{proof}

\StepTwo*
\label{lem:insertTransition*}

\begin{figure}[ptbh!]
    \centering
    \includegraphics[width=\textwidth, page=3]{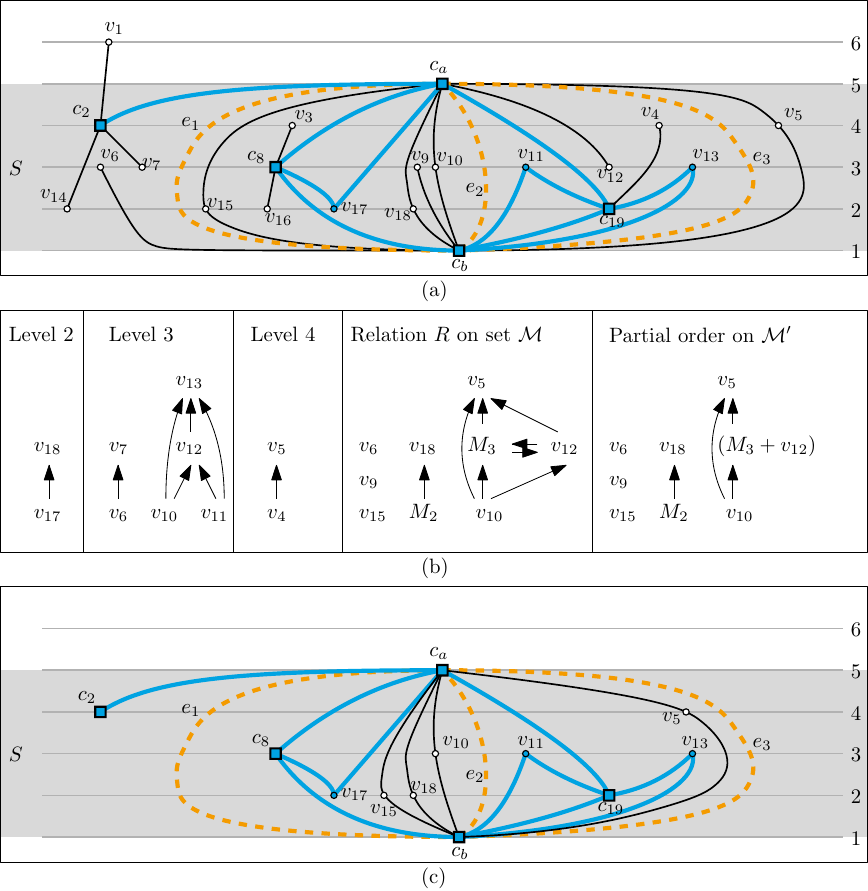}
  \caption{
  Like in all figures, filled square vertices belong to a vertex cover $C$ of the depicted graph and
    filled (round or square) vertices belong to the core of the shown drawing with respect to~$C$.
  (a)
   The drawing $\Lambda^*$; its visibility edges (all between $c_a, c_b$) are dashed.
Our goal is to insert the transition vertices of $c_a, c_b$ into the subdrawing $\Lambda_{\mathrm{core}}$ of $\Lambda^*$.
  The fixed components obtained by deletion of $c_a, c_b$ are
  $M_1$, consisting of $v_1, c_2, v_7, v_{14}$ and fixed to regions $r_1$ and $r_4$,
  $M_2$, consisting of $v_3, c_8, v_{16}, v_{17}$ and fixed to region $r_2$,
   and $M_3$, consisting of  $v_4, v_{11}, v_{13}, c_{19}$ and fixed to region $r_3$.
   The other (non-fixed) components are $v_5, v_6, v_9, v_{10}, v_{12}, v_{15}$, and $v_{18}$; note that these are all either transition vertices or leaves of $c_a, c_b$.
  (b)
  The first three boxes show the constraints of $\mathcal G$ between the vertices on levels 2, 3, and 4 (only vertices with constraints are shown).
  The remaining levels have no constraints.
  The fourth box shows the set $\mathcal{M}$ together with the relation $R$.
  Note that $M_1$ is not contained in $\mathcal{M}$ since it is fixed to $r_1$ and $r_4$.
  Further, $R$ contains a cycle.
  It is removed by the contraction step in which $\mathcal{M}'$ is obtained from $\mathcal{M}$, as depicted in the fifth box.
  (c)
  Consider the total order $v_9  \le_{\mathcal{M}'} M_2 \le_{\mathcal{M}'} v_6  \le_{\mathcal{M}'} v_{15} \le_{\mathcal{M}'} v_{18} \le_{\mathcal{M}'} v_{10} \le_{\mathcal{M}'} (M_3 + v_{12}) \le_{\mathcal{M}'} v_5$ on $\mathcal{M}'$.
  By removing all leaves from $\le_{\mathcal{M}'}$, we obtain the order $M_2 \le_{\mathcal{M}''} v_{15}  \le_{\mathcal{M}''} v_{18} \le_{\mathcal{M}''} v_{10} \le_{\mathcal{M}''} (M_3 + v_{12}) \le_{\mathcal{M}''} v_5$ on $\mathcal{M}''$.
  The figure shows the drawing $\Lambda^{\mathrm t}_{\mathrm{core}}$ obtained by inserting the transition vertices into $\Lambda_{\mathrm{core}}$ according to $\le_{\mathcal{M}''}$.
  Note that they are not inserted in the same places as in $\Lambda^*$. Further, we cannot simply insert the leaves alongside the transition vertices according to $\le_{\mathcal{M}'}$ since they can have  constraints to components in the outer regions $r_1$ and $r_4$ (as, in the example, $v_6$ to $M_1$).
  }
  \label{fig:transition}
\end{figure}

\begin{proof}
The terminology introduced in the following paragraphs is illustrated in \cref{fig:transition}(a).
Let $c_a,c_b\in C$ such that there exists at least one transition vertex adjacent to both $c_a$ and $c_b$, and let $e_1, e_2,\dots, e_j$ denote be the 
visibility edges joining~$c_a$ and~$c_b$ in the left-to-right order in 
which they appear in~$\Lambda^*$.
These edges partition the open horizontal strip~$S$ bounded by the lines $L_{\gamma(c_a)}$ and $L_{\gamma(c_b)}$ into regions $r_1,r_2,\dots r_{j+1}$ where~$e_i,1\le i\le j$ 
is on the boundary of $r_i$ and $r_{i+1}$.

Let $H_1, H_2, \dots, H_\ell$ be the connected components of the graph represented by $\Lambda^*$ after deleting $c_a, c_b$.
Note that these can contain vertices on the boundary 
or exterior to~$S$.
However, by planarity, if a component $H_i$ contains a vertex in a region~$r_{i'},2\le i'\le j$, then all of its vertices are located in~$r_{i'}$ in $\Lambda^*$.
Similarly, if a component $H_i$ contains a vertex exterior to all of the regions $r_2,r_3,\dots,r_j$ in $\Lambda^*$, then all of its vertices are located exterior 
to all of the regions $r_2,r_3,\dots,r_j$.
When~$H_i$ contains a vertex of~$C$ in~$r_{i'},1\le i'\le j+1$
in~$\Lambda^*$, then this vertex is drawn in~$r_{i'}$ in $\Lambda_\mathrm{core}$ as well; we say that~$H_i$ is \emph{fixed} to~$r_{i'}$.
Additionally, when~$H_i$ contains a vertex of~$C$ that is not located in the open strip~$S$, then we consider~$H_i$ to be \emph{fixed} to both~$r_1$ 
and~$r_{j+1}$ (representing the fact that~$H_i$ lies outside of the regions $r_2,r_3,\dots,r_j$).
In particular, every component~$H_i$ that contains at least two vertices also contains a vertex of~$C$ and is, thus, fixed to some region~$r_{i'}$.
Consequently, the only components that are not fixed are the transition vertices of $c_a$ and $c_b$ and the leaves of $c_a$ or $c_b$ (since $G$ has no isolated vertices).

The maximality of the set of visibility edges used to construct~$\Lambda^*$ implies that for each transition vertex~$v$ of~$c_a,c_b$, there exists an 
index~$i$ such that there is no fixed component interior to the region bounded by $e_i$ and the two incident edges of $v$,
in other words, $v$ is drawn in the vicinity of $e_i$ with respect to~$C$.

\subparagraph{Algorithm.}
To obtain $\Lambda^\mathrm{t}_{\mathrm {core}}$ from $\Lambda_\mathrm{core}$ we will draw every transition vertex of $c_a, c_b$ in the 
vicinity of one of the visibility edges of $c_a, c_b$.
To this end, we now define several sets and relations, for illustrations refer to \cref{fig:transition}(b).
Let $\mathcal{H}$ be the set that contains all non-fixed components (i.e. the transition vertices and leaves of $c_a$ and/or $c_b$), as well 
as the components that are fixed to one of the regions $r_2,r_3,\dots,r_j$.
For each region for which $\mathcal{H}$ contains multiple fixed components, we merge all of these components into a single \emph{fixed} subgraph 
and use $\mathcal{M}$ to denote the resulting set, which consists of one fixed subgraph for each of the regions $r_2,r_3,\dots,r_j$,
 as well as the transition vertices and leaves 
of $c_a$ and/or $c_b$.

We define a relation $R$ on $\mathcal{M}$ as follows:
$M_\alpha R M_\beta$ (with $M_\alpha\in \mathcal M$ and $M_\beta\in \mathcal M$) if and only if (1) there exists a vertex $v_\alpha$ in $M_\alpha$ 
and a vertex $v_\beta$ in $M_\beta$ such that $v_\alpha \prec_{\gamma(v_\alpha)} v_\beta$ or
(2) $M_\alpha$ and $M_\beta$ are fixed subgraphs with $M_\alpha$ fixed to $r_a$, $M_\beta$ fixed to $r_b$ and $a < b$ or
(3)  there exists a $M_\gamma\in \mathcal M$ such that $M_\alpha R M_\gamma$ and $M_\gamma R M_\beta$.

The relation~$R$ is by definition transitive, but it is not necessarily a partial order.
Nevertheless, we will now describe a strategy to derive a total ordering from~$R$ that will be useful for placing the transition vertices:
suppose $R$ contains a cycle $\delta$ of the form 
$M_1 R M_2 R \dots R M_i R M_1, i\ge 2$ of pairwise distinct elements $M_1, \dots, M_i \in \mathcal{M}$.
Assume towards a contradiction that one element of $\delta$, say $M_1$, consists of precisely one transition vertex $v$ of $c_a, c_b$.
In $\Gamma^*$, the incident edges of $v$ partition $S$ into a left and a right side.
An inductive argument shows that all components in the cycle $\delta$ are drawn entirely on the right side, including $M_1$; a contradiction.
Consequently, all elements of $\delta$ must be fixed subgraphs or leaves.
Without loss of generality we may assume that $M_1$ consists of at least two vertices ($\delta$ cannot consist exclusively of vertices lying on a 
common level $y$ since $\prec_y$ is a partial order).
This implies that $M_1$ is fixed, say, to $r_z$ ($2 \leq z \leq j$).
Similarly as before, inductive arguments show that all elements of $\delta$ have to lie to the right of $e_z$ and to the left of $e_{z+1}$, i.e., in $r_z$.
Therefore $M_1$ is the unique fixed component of $\mathcal{M}$ in $r_z$ (and also the unique fixed component in $\delta$).
We contract $\delta$ into a single subgraph inheriting all relations imposed by $R$ on $M_1, \dots, M_i$.
In this fashion, we successively eliminate all cycles and eventually obtain an ordering
on the resulting set $\mathcal{M}'$, which admits a linear extension~$\le_{\mathcal{M}'}$.
Observe that ordering of $\mathcal{M}'$ corresponding to $~\le_{\mathcal{M}'}$ is of the form 
\[S^1, F^2, S^2, F^3, S^3, \dots, F^j, S^j,\]
where $F^i$ denotes the fixed subgraph of $\mathcal{M}'$ in $r_i$ and each $S^w$ is a sequence of transition vertices and leaves of $c_a$ and/or $c_b$.
We remove all elements from $\mathcal{M}'$ that correspond to leaves and denote the resulting set by $\mathcal{M}''$.
Let $\le_{\mathcal{M}''}$ denote the total order on $\mathcal{M}''$ that corresponds to the  restriction of~$\le_{\mathcal{M}'}$ to $\mathcal{M}''$.
Observe that the ordering of $\mathcal{M}''$ corresponding to~$\le_{\mathcal{M}''}$ is of the form 
\[t^1_1, t^1_2, \dots, t^1_{i_1}, F^2, t^2_1, t^2_2, \dots, t^2_{i_2}, F^3, \dots, F^j, t^j_1, t^j_2, \dots, t^j_{i_j},\]
where $t^w_1, t^w_2, \dots, t^w_{i_w}$ is the restriction of $S^w$ to the transition vertices.
For each $2\leq w \leq j$ we draw $t^w_1, t^w_2, \dots, t^w_{i_w}$ in the vicinity of the visibility edge $e_w$ in this order (which can obviously be done without introducing crossings), see \cref{fig:transition}(c).
To obtain the desired drawing $\Lambda^\mathrm{t}_{\mathrm {core}}$, 
we repeat this procedure for every pair $c_a, c_b \in C$.
Clearly, all these steps can be performed in polynomial time.

We remark that it might seem tempting to draw the transition vertices, as well as the leaves according to~$\le_{\mathcal{M}'}$, instead of just drawing 
the transition vertices according to~$\le_{\mathcal{M}''}$. However, this will in general not produce a drawing with the desired properties as there might 
be leaves affected by constraints of components fixed to~$r_1$ and/or~$r_{j+1}$, which are not taken into account when computing 
$\le_{\mathcal{M}'}$ and $\le_{\mathcal{M}''}$, as illustrated in \cref{fig:transition}(c).

\subparagraph{Correctness.}

Suppose we have already carried out the above procedure for a set $C_2 \subset \binom{C}{2}$ of vertex cover vertex pairs $c_a, c_b$ and obtained a 
drawing $\Lambda_\mathrm{core}^\mathrm{C_2}$ that consists of $\Lambda_{\mathrm{core}}$ augmented by all transition vertices that belong to some pair in $C_2$.
Assume inductively that $\Lambda_\mathrm{core}^\mathrm{C_2}$ can be extended to a drawing, without loss of generality $\Lambda^*$, whose restriction 
to $G$ is a constrained level planar drawing of $\mathcal{G}$.
Consider a pair $c_a, c_b$ in $\binom{C}{2}\setminus C_2$ with at least one transition vertex  and let~$\Lambda_\mathrm{core}^\mathrm{C_2'}$ denote 
the drawing derived from $\Lambda_\mathrm{core}^\mathrm{C_2}$ by inserting the transition vertices of  $c_a, c_b$ according to the above procedure.
Our goal is to show that~$\Lambda_\mathrm{core}^\mathrm{C_2'}$ can be extended to a drawing whose restriction to $G$ is a constrained level planar drawing of $\mathcal{G}$.
To this end, we will construct a drawing~$\Lambda$ from~$\Lambda_\mathrm{core}^\mathrm{C_2}$ whose restriction to $G$ is a constrained level planar 
drawing of $\mathcal{G}$ and whose restriction to the graph represented by~$\Lambda_\mathrm{core}^\mathrm{C_2'}$ is~$\Lambda_\mathrm{core}^\mathrm{C_2'}$.

Let~$\Phi$ denote the simple closed curve corresponding to the boundary of the outer face of the subdrawing of~$\Lambda^*$ induced by the visibility 
edges joining~$c_a$ and~$c_b$ and the edges that are incident to some transition vertex of~$c_a,c_b$.
We use~$\Phi^-$ to denote the closed interior of~$\Phi$ and~$\Phi^+$ to denote the open exterior of~$\Phi$.
Note that all visibility edges joining~$c_a$ and~$c_b$ and all transition vertices of~$c_a,c_b$ are located in~$\Phi^-$ in~$\Lambda^*$ and, hence, 
all visibility edges joining~$c_a$ and~$c_b$ in $\Lambda_\mathrm{core}^\mathrm{C_2}$ are also located in~$\Phi^-$.

To construct~$\Lambda$ from $\Lambda_\mathrm{core}^\mathrm{C_2}$, we augment~$\Phi^+$ such that it looks exactly as in~$\Lambda^*$.
Additionally, we draw the fixed subgraphs~$F^2,F^3,\dots F^j$ exactly as in~$\Lambda^*$.
Let~$\Lambda'$ denote the resulting drawing.
To construct~$\Lambda$ from~$\Lambda'$, it remains to describe the placement of the  leaves of~$c_a$ and~$c_b$ (which we call \emph{leftover} leaves) that 
are located in~$\Phi^-$ in~$\Lambda^*$, but do not belong to any fixed subgraph $F^i$ (note that the leaves in~$\Phi^+$  in~$\Lambda^*$ and the leaves in 
the fixed subgraphs $F^i$ have already been drawn in~$\Lambda'$), as well as the placement of all the transition vertices of~$c_a,c_b$ (which are all located 
in~$\Phi^-$ in~$\Lambda^*$ and, thus, have not yet been drawn in~$\Lambda'$).
To this end, we draw those leftover leaves and transition vertices that occur in the sequence~$S^w, 2\le w\le j$ in the vicinity of the visibility edge~$e_w$ in the 
order given by~$S^w$ (which can be done without introducing crossings).

Since $\le_{\mathcal{M}'}~\subseteq~\le_{\mathcal{M}''}$, the restriction of~$\Lambda$ to the graph represented by~$\Lambda_\mathrm{core}^\mathrm{C_2'}$ 
corresponds to~$\Lambda_\mathrm{core}^\mathrm{C_2'}$, as desired.
Hence, it remains to argue that the restriction~$\Lambda_G$ of~$\Lambda$ to $G$ is a constrained level planar drawing of $\mathcal{G}$.
By construction, $\Lambda_G$ is a level planar drawing of~$\mathcal G$, so assume that some constraint $u\prec_{\gamma(u)} v$ is violated in~$\Lambda_G$, 
i.e., $u$ is drawn to the right of~$v$ on~$L_{\gamma(u)}$ (instead of to the left, as required).
By definition of~$\Lambda^*$, the restriction of~$\Lambda'$ to its non-visibility edges is a constrained level planar drawing of the corresponding subgraph of~$\mathcal G$.
Consequently, one of $u$ or $v$ is a vertex of~$\Lambda$ that does not belong to~$\Lambda'$, i.e., a leftover leave or transition vertex of $c_a$ 
and/or $c_b$, which lies in~$\Phi^-$ in~$\Lambda^*$.
Assume without loss of generality that this vertex is~$v$ (the other case is symmetric).
The definition of $\le_{\mathcal{M}'}$ implies that~$u$ lies in~$\Phi^+$ and, by assumption, to the right of~$\Phi^-$, which contains~$v$.
However, by definition of~$\Lambda'$, the vertex~$u$ also lies to the right of~$\Phi^-$ in~$\Lambda^*$.
Since the vertex~$v$ also lies in~$\Phi^-$ in~$\Lambda^*$, it follows that the constraint $u\prec_{\gamma(u)} v$ is also violated in~$\Lambda^*$; 
a contradiction to the definition of~$\Lambda^*$.
So indeed, $\Lambda$ (and, thus, $\Lambda_\mathrm{core}^\mathrm{C_2'}$) satisfies all desired properties, proving the correctness of the algorithm.
\end{proof}

\ChannelNumber*
\label{claim:channelnumber_bounded*}

\begin{claimproof}
Assume there are two channels $(v,r,R_1),(v,r,R_2)$ with $R_1\neq R_2$ and such that the unique cell in~$R_1$ that contains~$v$ and the unique cells in~$R_2$ that contains~$v$ are identical.
Then there exist two cells $r_a,r_b\in R_1\cap R_2$ with $r_a$ above $r_b$
such that it is possible to draw two y-monotone curves $c^1,c^2$ that do not cross each other and that do not intersects any edge or vertex of $\Lambda^\mathrm{t}_{\mathrm {core}}$ and where the endpoints $p^1_a,p^1_b$ and  $p^2_a,p^2_b$ of $c^1$ and $c^2$, respectively, satisfy the following properties:
\begin{itemize}
\item $p^1_a$ belongs to the lower boundary of $r_a$ and the upper boundary of some cell $r_a^1\in R_1\setminus R_2$;
\item $p^2_a$ belongs to the lower boundary of $r_a$ and the upper boundary of some cell $r_a^2\in R_2\setminus R_1$;
\item $p^1_b$ belongs to the upper boundary of $r_b$ and the lower boundary of some cell $r_b^1\in R_1\setminus R_2$;
\item $p^2_b$ belongs to the upper boundary of $r_b$ and the lower boundary of some cell $r_b^2\in R_2\setminus R_1$.
\end{itemize}
By construction, the simple closed curve formed by $c^1$ and $c^2$ and the line segments $p^1_ap^2_a$ and $p^1_bp^2_b$ intersects no vertex or edge of $\Lambda^\mathrm{t}_{\mathrm {core}}$, has $v$ in its exterior, and contains at least one vertex~$v'\neq v$ of $\Lambda^\mathrm{t}_{\mathrm {core}}$ in its interior.
However, the connectivity of $\Lambda^\mathrm{t}_{\mathrm {core}}$ (cf.\ \cref{lem:size-of-core}) implies that $\Lambda^\mathrm{t}_{\mathrm {core}}$ contains a path from $v$ to $v'$; a contradiction.

It follows that for every pair of channels $(v,r,R_1),(v,r,R_2)$ with $R_1\neq R_2$, the unique cell in~$R_1$ that contains~$v$ and the unique cell in~$R_2$ that contains~$v$ are distinct.
Consequently, we obtain $|U| \leq |R|\cdot \sum_{v \in C} \alpha(v)$, where $\alpha(v)$ denotes the number of cells incident to~$v$ for each $v\in C$.
Recall that each cell is y-monotone and bounded by up to two horizontal segments or rays and up to two y-monotone curves.
Now, let $v\in C$, let $\alpha^{|}(v)$ denote the number of cells incident to~$v$ for which~$v$ is incident to at least one y-monotone curve of the cell boundary, and let $\alpha^{-}(v)$ denote the number of the remaining cells incident to~$v$ (where $v$ is incident only to horizontal segments or rays of the cell boundary).
Observe that $\alpha^{-}(v)\le 2$.
Additionally, each cell~$r$ can contain at most $4$ vertices of $C$ that are incident to y-monotone curves along the boundary of~$r$.
Putting it all together, we obtain
\[|U| \leq |R|\cdot \sum_{v \in C} \alpha(v) = |R|\cdot (\sum_{v \in C} \alpha^{|}(v) + \sum_{v \in C} \alpha^{-}(v)) \leq |R|\cdot (4 \cdot |R| + |C| \cdot 2) \in \mathcal O(k^2),\]
where the final inclusion uses $|C|=k$ and $|R|\in \mathcal O(k)$.
\end{claimproof}

\TraversalSequence*
\label{cl:compat-traversal-exists*}

\begin{claimproof}
Let $\prec^{\mathrm{e}}$ be a compatible edge ordering for the restriction of $\Lambda^*$ to its nonhorizontal edges and let $\prec^l$ be the restriction of $\prec^{\mathrm{e}}$ to the edges that are incident to leafs.
We will now label every used channel with a set of numbers corresponding to the edges that use it in $\Lambda^*$:
let $c$ be a used channel. We label $c$ with the number $i$ if and only if  the $i$th edge in $\prec^l$ uses $c$.
Note that a channel can be assigned multiple labels.
Let $m$ be the number of leaves. For every $1 \leq i \leq m$, we define $\mathcal{U}_i$ 
to be all channels labeled with two (not necessarily distinct) $j, j'$ fulfilling $j \leq i \leq j'$; see \cref{fig:traversal} for examples.
Let $\mathcal U=(\mathcal{U}_1, \mathcal{U}_2, \dots, \mathcal{U}_m)$ denote the resulting sequence.

\subparagraph{$\mathcal U$ is a traversal sequence.}
We will now show that the sequence $\mathcal U$ is a traversal sequence.
To show Property~\ref{T:unique_y_coord},
let $i\in [m]$ and let $c=(v,r,R)\in \mathcal{U}_i$ and $c'=(v',r',R')\in\mathcal{U}_i$
		with $v\neq v'$.
		Assume towards a contradiction that the line $L_{\gamma(v)}$ intersects the interior of the union of $R'$.
		Let $r_i$ denote a cell of $R'$ that is intersected by $L_{\gamma(v)}$ and let $e$ denote the edge that corresponds to the left boundary of~$r_i$.
		Assume without loss of generality that~$r_i$ lies to the right of~$v$.
		Note that every edge incident to~$v$ precedes~$e$ in~$\prec^{\mathrm{e}}$.
		Moreover, every edge that intersects the interior of~$r_i$ succeeds~$e$.
		Hence, every edge using~$c$ precedes every edge using~$c'$ in~$\prec^l$.
		Since~$c\in \mathcal U_i$, it contains a label that is at least $i$ and, hence, it follows that all labels in~$c'$ are strictly larger than~$i$;
		a contradiction to the fact that~$c'\in \mathcal U_i$.
		Therefore, Property~\ref{T:unique_y_coord} is satisfied.

To show  Property~\ref{T:interval}, let $c$ be a (used) channel that occurs in~$\mathcal U$ and let $a$ be the lowest index with 
$c \in \mathcal{U}_a$ and let $b$ be the highest index with $c \in \mathcal{U}_b$.
Then for every $a \leq i \leq b$, $c \in \mathcal{U}_i$ by construction and  Property~\ref{T:interval} is fulfilled.
Therefore, $\mathcal U$ is indeed a traversal sequence.

\subparagraph{$\mathcal U$ is compatible with $\Lambda^*$.}
We will now show the compatability of $\mathcal U$ with $\Lambda^*$.
Let $c$ be a used channel. Then there exists at least one edge $e$ using $c$ in $\Lambda^*$. Let $i$ be the index of $e$ in $\prec^l$.
Then, by construction, $c \in \mathcal{U}_i$.
Moreover, by construction, $\mathcal U$ does not contain any channels that are not used.
Therefore,  Property~\ref{c1:usedchannelsappear} is satisfied.

To establish Property~\ref{C2a:real_before},
let $e, e'$ be two edges that are incident to leaves with $e \prec^l e'$, let $c$ be the channel used by $e$ in $\Lambda^*$ and let $c'$ be the channel used by $e'$ in $\Lambda^*$.
Let further $i, i'$ be the indices of $e,e'$, respectively, in the ordering corresponding to $\prec^l$.
Then $i < i'$ and $c \in \mathcal{U}_{i}$ and $c' \in \mathcal{U}_{i'}$ by construction and, therefore,  Property~\ref{C2a:real_before} is satisfied. 

To establish Property~\ref{C2b:exclusivenes},
let $c_1,c_2\in U_{\mathrm{used}}$ such that for every edge $e_1$ using $c_1$ and for every edge $e_2$ using~$c_2$, we have $e_1 \prec^{\mathrm{e}} e_2$.
Let $i$ be the largest index of an edge in $c_1$ with respect to $\prec^l$
and let $j$ be the smallest index of an edge in $c_2$ with respect to $\prec^l$.
Note that $i<j$ and, by construction, the interval in which $c_1$ is active ends with $\mathcal{U}_i$ and the interval in which $c_2$ is active starts with $\mathcal U_j$.
Therefore, there is no index at which both $c_1$ and $c_2$ are active and $c_1$ is active before $c_2$, which establishes Property~\ref{C2b:exclusivenes}.

To establish Property~\ref{C2c:real_after},
consider two used channels $c_1$, $c_2$ such that $c_2$ is being used by an edge $e$ that succeeds all edges that use $c_1$ in $\prec^{\mathrm{e}}$.
Let $i$ be the index of~$e$ with respect to~$\prec^l$. 
Then $c_2 \in  \mathcal{U}_i$ and $c_1 \notin \mathcal{U}_i$ by construction (and $c_1$ is active before $c_2$ by Property~\ref{C2a:real_before}), which establishes Property~\ref{C2c:real_after}.
Therefore, $\mathcal U$ is indeed compatible with~$\Lambda^*$.

\subparagraph{Bounding the length of the sequence.}
Note that if there are two consecutive sets $\mathcal U_i,\mathcal U_{i+1}$ in $\mathcal U$ with $\mathcal U_i=\mathcal U_{i+1}$, then we can remove one of them and the resulting sequence~$\mathcal U'$ is still a traversal sequence that is compatible with~$\Lambda^*$, e.g, in \cref{fig:traversal} the set $\mathcal U_8$ can be removed.
Repeat this process to eliminate all consecutive pairs of identical sets $\mathcal U_i$ and denote the resulting sequence by~$\mathcal U''$.
Since the number of channels in~$\mathcal U''$ is $\mathcal O(k^2)$ by \cref{claim:channelnumber_bounded}
and the sets containing a channel form an interval by Property~\ref{T:interval}, it follows that the length of~$\mathcal U''$ is $\mathcal O(k^2)$.

\subparagraph{Algorithm.}
So far, we have shown that there is a constant $a$ such that there is a traversal sequence compatible with $\Lambda^*$ of length at most $ak^2$.
Moreover, we know that there is a constant $b$ such that $|U|\le bk^2$ by \cref{claim:channelnumber_bounded}.
Every traversal sequence of length at most $ak^2$ is uniquely described by specifying for each channel $c\in U$ at which index in $\{1,2,\dots,ak^2\}$ its interval (cf.\ Property~\ref{T:interval}) starts and at which index it ends (when the starting index  is larger than the ending index, $c$ does not appear in the sequence).
Hence, there are at most $(ak^2)^{(2bk^2)}\le 2^{\mathcal O(k^2\log k)}$ candidate sequences.
Testing whether a candidate sequence is a traversal sequence is easily done in polynomial time and the existence of~$\mathcal U''$ implies that at least one of the sequences is a traversal sequence that is compatible with $\Lambda^*$.
\end{claimproof}

\ChannelsUniqueInEachSet*
\label{claim:channels-unique-in-each-set*}

\begin{claimproof}
Assume towards a contradiction that $\mathcal U_i$ contains
two channels $c=(w,r,R)$ and $c'=(w,r',R')$ that can be used by~$v$.
Then there are two cells $a\in R\setminus R'$ and $a'\in R'\setminus R$ such that there is a horizontal line intersecting both $a$ and $a'$.
This implies that the precondition of Property~\ref{C2b:exclusivenes} is fullfilled and
the conclusion of Property~\ref{C2b:exclusivenes} implies that $c,c'$ are not both contains in $\mathcal U_i$, which yields the desired contradiction.
\end{claimproof}

\InsertionSequenceExists*
\label{cl:insertion-sequence-exists*}

\begin{claimproof}
Suppose, inductively, we have obtained insertion sequences
 $\mathcal{Q}_0, \mathcal{Q}_1,\dots,\mathcal{Q}_q$
 for $i$, $s$ and $\mathcal U$
 such that for each of them
the interval and the dominance properties are fulfilled
and
such that $\mathcal{Q}_k$ is a prefix of $\mathcal{Q}_{k+1}$ for  all $0 \leq k \leq q-1$.
Note that if $q = 0$, then $\mathcal{Q}_0 = ()$ 
is an insertion sequence with the desired properties.
Let 
$\mathcal{Q}_q = (Q_1, Q_2, \dots, Q_q)$, $0 \leq q < |V_i \cap V_{=1}|$.
Our goal is to construct the sequence $\mathcal{Q}_{q+1}$.

\subparagraph{Existence of choosable vertices.}
We first show that there is a choosable vertex~$v$ with regard to $\mathcal{Q}_q$.
To this end, it suffices to show that there exists at least one pair $(v,j)$ such that $(Q_1, Q_2, \dots, Q_q, (v,j))$ is an insertion sequence for $i$, $s$, and $\mathcal U$
(note that this does not necessarily imply that~$v$ is choosable, but its existence implies that there is a choosable vertex, as desired).

Let $v$ be the left-most leaf on level $i$ in $\Lambda^*$ that is not contained in $\mathcal{Q}_q$ and let $c$ be the channel used by $v$ (in $\Lambda^*$). 
Let $j$ be the largest index such $c \in \mathcal{U}_j$, which exists by Property~\ref{c1:usedchannelsappear}.
Note that if~$q=0$, the definition of~$v$ implies that $(v,j)$ can be appended to the (empty) sequence~$\mathcal Q_q$ to obtain an insertion sequence for $i$, $s$, and~$\mathcal U$, as desired.
So assume otherwise, i.e., $q \ge 1$.
Let $Q_q=(v_q,j_q)$.

Our main goal is now to show that $j_q\le j$, which is a necessary condition for appending  $(v,j)$ to $\mathcal Q_q$.
Towards a contradiction, assume $j < j_q$.
Let $j^*$ be the largest index that occurs in $\mathcal Q_q$ and where $j^* \leq j$
(if no such index exists, we define $j^*=0$).
Let $q^*$ be the largest index such that~$Q_{q^*}$ contains~$j^*$
(if~$j^*=0$, we define $q^*=0$).

We will now argue that~$(v,j)$ can be appended to $\mathcal Q_{q*}$ to obtain an insertion sequence for $i$, $s$ and $\mathcal U$,
which we will then use to argue that $\mathcal Q_{q}$ does not satisfy the interval property, yielding the desired contradtion.

So let $\overline{\mathcal{Q}_{q^* +1}}$ be the sequence obtained from $\mathcal{Q}_{q^*}$ by appending $(v, j)$.
Clearly, Properties \ref{i1}, \ref{i2}, \ref{i3}, and \ref{i5} are satisfied for the sequence $\overline{\mathcal{Q}_{q^* +1}}$.
To see that  Property~\ref{i4} also holds,
assume towards a contradiction that there is a vertex $\overline{v} \in V_i \cap V_{=1}$ with $\overline{v} \prec_i v$ that is not contained in $\overline{\mathcal{Q}_{q^* +1}}$.
By the definition of~$v$, the vertex $\overline{v}$ is contained in $\mathcal{Q}_q$.
Without loss of generality, assume that $\overline{v}$ is the first vertex in $\mathcal{Q}_q$ with these properties.
Let $(\overline{v},\overline{j})\in \mathcal{Q}_q$.

The constraint $\overline{v} \prec_i v$ implies together with Property~\ref{C2a:real_before}
that there exists an index $j'\le j$
such that the
channel $\overline c$ used by $\overline{v}$ (in $\Lambda^*$) is contained in
$\mathcal U_{j'}$.
To obtain the desired contradiction (to the existence of $\overline v$), we will now show that 
$(\overline{v},j')$ is can be appended to 
$\mathcal{Q}_{q^*}$ to obtain an insertion sequence for $i$, $s$ and $\mathcal U$.
Since $\overline{v}\in \mathcal{Q}_q$ and $\overline{v}\notin \overline{\mathcal{Q}_{q^*}}$, it follows that $j^*<\overline{j}$ (by definition of $j^*$).
Further, since (again by definition of $j^*$) there is no index~$\tilde j$ that occurs in $\mathcal{Q}_q$ and where $j^*<\tilde j\le j$, it follows that $j < \overline{j}$.
Due to the dominance property for $\overline v$ in $\mathcal{Q}_q$,
$\overline j$ is at most as large as the largest index~$j_c$  for which $\overline c\in \mathcal U_{j_c}$.
To summarize, $j'\le j<\overline j\le j_c$.

By Property~\ref{T:interval} and since~$\overline c\in \mathcal U_{j'}$ and $\overline c\in \mathcal U_{j_c}$ and $j'\le j<\overline j\le j_c$, it follows that $\overline c\in \mathcal U_j$.
Then, since $\overline v$ is the first vertex in $\mathcal{Q}_q$ with $\overline{v} \in V_i \cap V_{=1}$ with $\overline{v} \prec_i v$ that is not contained in $\overline{\mathcal{Q}_{q^* +1}}$,
it follows that $(\overline v,j)$ can be appended to $\mathcal{Q}_{q^*}$
without violating Property~\ref{i4} to obtain an insertion sequence for $i$, $s$, and $\mathcal U$.
Moreover, note that $j<j_{q^*+1}$ where $Q_{q^*+1}=(v_{q^*+1}, j_{q^*+1})$ and, hence $v_{q^*+1}$ is not choosable with regard to $\mathcal{Q}_{q^*}$, which yields a contradiction to the interval property for $v_{q^*+1}$ in $\mathcal{Q}_{q}$.
Therefore, the vertex $\overline{v} \in V_i \cap V_{=1}$ with $\overline{v} \prec_i v$ that is not contained in $\overline{\mathcal{Q}_{q^* +1}}$ cannot exists an, hence, Property~\ref{i4} is satisfied by $\overline{\mathcal{Q}_{q^* +1}}$, which, finally, shows that $\overline{\mathcal{Q}_{q^* +1}}$ is an insertion sequence $i$, $s$ and $\mathcal U$.

Once again, note that $j<j_{q^*+1}$ where $Q_{q^*+1}=(v_{q^*+1}, j_{q^*+1})$ and, hence $v_{q^*+1}$ is not choosable with regard to $\mathcal{Q}_{q^*}$, which yields a contradiction to the interval property for $v_{q^*+1}$ in $\mathcal{Q}_{q}$.
Therefore $j_q\le j$.
The fact that $j_q\le j$ can be combined with the definition of~$v$ to obtain that
$(Q_1, Q_2, \dots, Q_q, (v,j))$ is indeed an insertion sequence for $i$, $s$, and $\mathcal U$, as claimed, and, thus, there exists at least one choosable vertex with regard to $\mathcal Q_q$.

\subparagraph{Interval property.}
Now let $v$ be a choosable vertex together with an smallest index $j$ such that $\mathcal{Q}_{q+1} = (Q_1, Q_2, \dots, Q_q, Q_{q+1} = (v,j))$ 
is an insertion sequence for $i$, $s$, and $\mathcal U$.
We will now show that it satisfies the interval property.
Let $v' \in V_i \cap V_{=1}$ be a vertex. We will show that the interval property is fulfilled for $v'$.
We distinguish four cases.
\subparagraph{Case~1:} $v' = v$. Then, $v'$ has been choosable in $\mathcal{Q}_q$ in an interval ending with $q$ and $v\in Q_{q+1}$, thus the interval 
property is fulfilled.
\subparagraph{Case~2:} $v'\in \mathcal Q_q$. Hence, $v'\in \mathcal Q_{q+1}$ and, thus, it is not choosable for $\mathcal{Q}_{q+1}$, and the interval property is fulfilled for $\mathcal{Q}_{q+1}$ 
since it was already fulfilled in $\mathcal{Q}_q$.
\subparagraph{Case~3:} The interval in which $v'$ was choosable in $\mathcal{Q}_q$ is empty.
Then the interval in which~$v'$ is choosable in~$\mathcal Q_{q+1}$ has length at most one and if it has length exactly one, it is choosable only with regard to~$\mathcal Q_{q+1}$, hence, the interval property is fulfilled.
\subparagraph{Case~4:} $v'$ is not contained in $\mathcal{Q}_q$ and the interval in which $v'$ was choosable in $\mathcal{Q}_q$ is nonempty and ends with $q$.
We need to show that~$v'$ is choosable with regard to $\mathcal Q_{q+1}$.

There exists a channel $c'$ in $U_j$ that can be used by $v'$ since we could have appended $v'$ with the same index as we included $v$.
Further, there exists no constraint of the form $v' \prec_i v$ or $v \prec_i v'$, otherwise one of them would have not been choosable with regard to~$\mathcal Q_q$ as appending it would have 
violated Property~\ref{i4} of the resulting insertion sequence.
Let $\overline{\mathcal{Q}_{q+2}}$ denote the sequence obtained by appending $(v', j)$ to $\mathcal{Q}_{q+1}$.
We claim that $\overline{\mathcal{Q}_{q+2}}$ is an insertion sequence.
Indeed, Property~\ref{i1} is fulfilled, since $v'$ was not in  $\mathcal{Q}_{q+1}$. 
Property~\ref{i2} is fulfilled for $\overline{\mathcal{Q}_{q+2}}$ since $v'$ and $v$ both were appended with index~$j$.
We already stated that $v'$ can use a channel~$c'$ in $U_j$, therefore, Property~\ref{i3} is fulfilled.
Property~\ref{i4} is fulfilled since~$v'$ was choosable with regard to~$\mathcal Q_{q+1}$.
Property~\ref{i5} is fulfilled for $\overline{\mathcal{Q}_{q+2}}$ since the unique usable channel $c'$ in $U_j$ that is used by $v'$ in $\overline{\mathcal{Q}_{q+2}}$ does not change.
This shows that $v'$ is still choosable with regard to $\mathcal{Q}_{q+1}$, and the interval property is fulfilled.

So in any case, the interval property for $\mathcal{Q}_{q+1}$ is fulfilled.

\subparagraph{Dominance property.}
It is left to show that the dominance property is fulfilled for $\mathcal Q_{q+1}$.
For all vertices in $\mathcal Q_{q+1}$ but~$v$, the dominance property carries over from~$\mathcal Q_q$.

Let $v_\ell' \in V_i \cup V_{=1}$ and let $e$ ($e_{\ell'}$) be the edge incident to $v$ ($v_{\ell'}$) with $e \prec^l_i e_{\ell'}$ or $v  = v_{\ell'}$. 
We have to show that $j \le j'$, where $j'$ is the maximum index such that the channel $c'$ used by $v_{\ell'}$ (in $\Lambda^*$) is in $\mathcal U_{j'}$.
We distinguish two cases:

\subparagraph{Case 1:} $v = v_{\ell'}$. 
Assume towards contradiction that $j' < j$ and let $j^*$ be the largest index such that $j^*$ is present in $\mathcal{Q}_{q}$ and $j^* \leq j'$
(if such an index does not exists, we define $j^*=0$).
Further, let $q^*$ be the largest index such that $Q_{q^*}$ contains index $j^*$ (and $q^*=0$ if $j^*=0$).

Let $\overline{\mathcal{Q}_{q^* + 1}}$ be the sequence obtained from $\mathcal{Q}_{q^*}$ by appending $(v, j')$.
We will show $\overline{\mathcal{Q}_{q^* + 1}}$ is an insertion sequence for $i$, $s$, and $\mathcal U$,
 which gives us a contradiction to the interval property of $v_{q^*+1}$, where $v_{q^*+1}\in Q_{q^*+1}$ (since it is not choosable with regard to $\mathcal{Q}_{q^*}$).
 
Clearly, Properties~\ref{i1}, \ref{i2}, \ref{i3}, and \ref{i5} are fulfilled for $\overline{\mathcal{Q}_{q^* + 1}}$.
It remains to establish Property~\ref{i4}.
Since $v$ is choosable in~$\mathcal Q_q$, all vertices $w\in V_i \cup V_{=1}$ with $w \prec_i v$ are already in $\mathcal{Q}_q$ and, thus, the dominance property is fulfilled for them in $\mathcal{Q}_q$. Let $w$ be such a vertex, let $j_w$ be the index with which $w$ is contained in $\mathcal{Q}_q$ and let $e_w$ be the edge incident to $w$. Clearly, $e_w \prec^l_i e$.
Hence, by the dominance property of $w$, it follows $j_w \leq j'$ (and also $j_w \leq j^*$ since $j_w$ is present in $\mathcal{Q}_{q}$) and therefore, $w$ is contained in $\mathcal{Q}_{q^*}$.
Thus, Property~\ref{i4} is also fulfilled and $\overline{\mathcal{Q}_{q^* +1}}$ is indeed an insertion sequence.
This is a contradiction to the interval property of $v_{q^*+1}$, where $v_{q^*+1}\in Q_{q^*+1}$ since it is not choosable with regard to $\mathcal{Q}_{q^*}$.
Thus, $j' \leq j$, as desired.

\subparagraph{Case 2:} $e \prec^l_i e_{\ell'}$.
Assume towards contradiction that $j' < j$.
Because of Case 1, we know that $j$ is at most as large as the largest index of the channel $c$ used by $v$ (in $\Lambda^*$) and, hence~$j'$ is strictly smaller than this index.
 Also because of Property~\ref{C2a:real_before}, we know that there exists an index $\tilde j$ with $\tilde j \leq j'$ and $c \in \mathcal{U}_{\tilde j}$, thus, by Property~\ref{T:interval}, we know that $c \in \mathcal{U}_{j'}$. 

Now let $j^*$ be the largest index such that $j^*$ is present in $\mathcal{Q}_{q}$ and $j^* \leq j'$
(if such an index does not exists, we define $j^*=0$).
Further, let $q^*$ be the largest index such that $Q_{q^*}$ contains index $j^*$ (and $q^*=0$ if $j^*=0$).
Let $\overline{\mathcal{Q}_{q^* + 1}}$ be the sequence obtained from $\mathcal{Q}_{q^*}$ by appending $(v, j')$.
We will show $\overline{\mathcal{Q}_{q^* + 1}}$ is an insertion sequence for $i$, $s$, and $\mathcal U$,
which gives us the desired
contradiction to the interval property of $v_{q^*+1}$, where $v_{q^*+1}\in Q_{q^*+1}$ (since it is not choosable with regard to $\mathcal{Q}_{q^*}$).

Clearly, Properties~\ref{i1}, \ref{i2}, \ref{i3}, and \ref{i5} are fulfilled for $\overline{\mathcal{Q}_{q^* + 1}}$.
It remains to establish Property~\ref{i4}.
Since $v$ is choosable in~$\mathcal Q_q$, all vertices $w\in V_i \cup V_{=1}$ with $w \prec_i v$ are already in $\mathcal{Q}_q$ and, thus, the dominance property is fulfilled for them. 
Let $w$ be such a vertex, let $j_w$ be the index with which $w$ is contained in $\mathcal{Q}_q$ and let $e_w$ be the edge incident to $w$. Clearly, $e_w \prec^l_i e \prec^l_i e_{\ell'}$.
Hence, by the dominance property of $w$, it follows $j_w \leq j'$ (and also $j_w \leq j^*$ since $j_w$ is present in $\mathcal{Q}_{q}$) and therefore, $w$ is contained in $\mathcal{Q}_{q^*}$.
Thus, Property~\ref{i4} is also fulfilled and $\overline{\mathcal{Q}_{q^* +1}}$ is indeed an insertion sequence.
This is a contradiction to the interval property of $v_{q^*+1}$, where $v_{q^*+1}\in Q_{q^*+1}$ since it is not choosable with regard to $\mathcal{Q}_{q^*}$.
Thus, $j' \leq j$, as desired.

Therefore, the dominance property is satisfied.

\subparagraph{Algorithm.}
We have shown that we can generate  $\mathcal{Q}_{|V_i \cap V_{=1}|}$ by starting with an empty sequence $\mathcal{Q}_0 = ()$ and appending choosable 
vertices until all leaves of level $i$ are contained in the insertion sequence (the insertion sequence illustrated in \cref{fig:insertion} was created by means of this algorithm). Clearly, testing whether a vertex is choosable can be done in polynomial 
time, therefore we can generate $\mathcal{Q}_{|V_i \cap V_{=1}|}$ in polynomial time as well.
\end{claimproof}

\GenerateDrawingFromEarOrientation*
\label{cl:generate-drawing-from-ear-orientation*}

\begin{claimproof}
For every level $i\in [h]$, draw the ears in the unique way described by $s^i$, which can be done without creating any crossings.
Let $\Lambda_\mathrm{core}^{\mathrm{t, e}}$ denote the resulting drawing.

We now add the leaves to $\Lambda_\mathrm{core}^{\mathrm{t, e}}$ according to the insertion sequences as follows.
We consider all indices $\ell \in [m]$ in increasing order.
To process an index $\ell$, we consider all levels $i\in [h]$ in increasing order.
To process a level $i$,
we consider the indices $q\in [q^i]$ in increasing order.
To process the index $q$, 
let $Q^i_q = (v^i_q, j^i_q)$ be the $q$-th entry of the insertion sequence $\mathcal{Q}^i$. If $j^i_q = \ell$,
let $c=(w,r,R)$ be unique channel
usable by $v$ in $\mathcal{U}_\ell$.
If $R$ consists of a single cell (namely, $r$) that contains transition vertices, we just \emph{assign} $v$ to $r$, without drawing it yet.
Assigned vertices will be drawn in a post-processing step.
Otherwise, i.e., if~$r$ is not a cell that contains transition vertices
, we draw a (nearly horizontal) line from $w$ to the right boundary of the union of $R$, then traverse along upwards or downwards  along this boundary until we reach the y-coordinate $\gamma(v)$.

Assume that we have already performed some number of steps of the algorithm while maintaining the following invariants:
\begin{enumerate}
\item The current drawing is crossing-free;
\item For each level~$i$,
 the leaves already present  in the drawing (which do not include the assigned leaves)
are in the same order as in $\mathcal{Q}^i$; and
\item Let $v$ be a leave that is already present or assigned and let $Q^i_\ell=(v,j)\in \mathcal Q^i$.
Further, let $c=(w,r,R)$ be the unique channel usable by $v$ in $\mathcal U_j$.
Then~$v$ is either drawn in~$c$ or assigned to~$r$.
\end{enumerate}
Let $(v,j)$ be the tuple that is processed in the next step and let~$c=(w,r,R)$ be the unique channel in~$\mathcal U_j$ that can be used by~$v$.
If $R$ consists of a single cell (namely, $r$) that contains transition vertices, the invariants are obviously maintained.
Otherwise, if~$r$ does not contain transition vertices, the curve along which the algorithm draws the edge incident to~$v$ is not crossed by any of the previously drawn edges by Property~\ref{T:unique_y_coord} and Property \ref{C2b:exclusivenes} of the traversal sequence and, hence, Invariant~1 is maintained.
For Invariant~2, assume towards a contradiction that there is an already inserted leaf~$w$ that is to the right of~$v$.
By construction, it is in a cell that is to the right of the one of $v$. This is a contradiction to Property~\ref{C2b:exclusivenes}.
Clearly, Invariant~3 is satisfied.

Hence, the invariants are maintained.
Let~$\Lambda$ denote the resulting drawing.
It remains to draw the unassigned leaves.
For each level $i$ and for each cell~$r$ with transition vertices on level $i$, we simply perform a topological sort with respect to~$\prec_i$ on all vertices in~$r$ (which are assigned leaves and transition vertices)
and then place them according to the obtained order (possibly rearranging the order of the transition vertices in~$r$), which is easy to do without introducing crossings.
Let~$\Lambda'$ denote the final drawing, which is cross-free.

It remains to show that in~$\Lambda'$ all constraints are respected.
By construction, this is the case for the vertices in $\Lambda_\mathrm{core}^{\mathrm{t, e}}$.

So let $v$ be a leaf on level $i$, let $c$ be the channel used by $v$ in~$\Lambda'$.
All constraints between $v$ and other leaves are respected as all leaves on level $i$ appear in the same order as in~$\mathcal Q^i$ up to permutations in transition vertex cells according to topological orderings with respect to~$\prec_i$.
Constraints between $v$ and non-leafs are satisfied by Property~\ref{i5}.

Clearly, all steps of the algorithm can be carried out in polynomial time.
\end{claimproof}

\SumStepThree*
\label{cl:summaryStep3*}

\begin{claimproof}
We begin by augmenting $\Lambda_\mathrm{core}^{\mathrm{t}}$ 
as described in the beginning of the proof of \cref{lem:insertLeftovers}, which is easy to do in polynomial time.
The decomposition into cells can also easily be constructed in polynomial time.
We then use 
\cref{cl:compat-traversal-exists}
to construct a set of $2^{\mathcal O(k^2\log k)}$ candidate traversal sequences of length at most $\mathcal O(k^2)$ that is guaranteed to contain a traversal sequence that is compatible with~$\Lambda^*$
in
$2^{\mathcal O(k^2\log k)}n^{\mathcal O(1)}$
time.

For each of these candidate sequences, we now try to construct a valid ear orientation together with an insertion sequence for each level.
To do this for a given level~$i$, we enumerate all (at most $2^{2k}$) valid ear orientations
for level~$i$
in
$2^{2k}n^{\mathcal O(1)}$ time.
For each of these valid ear orientations, 
 we then try to apply the algorithm corresponding to \cref{cl:insertion-sequence-exists} to obtain an insertion order, which takes polynomial time per ear orientation.
If we succeed in finding a valid ear orientation together with an insertion sequence for each level, we now try to apply the algorithm corresponding to \cref{cl:generate-drawing-from-ear-orientation} to construct the desired drawing, which takes polynomial time.
If we fail at any point along the way, we just continue with the next candidate traversal sequence.

Since our candidate set contains a traversal sequence that is compatible with~$\Lambda^*$, we are guaranteed to eventually obtain the desired drawing.
In particular, when applying our strategy to the compatible sequence, we are guaranteed to obtain a valid ear orientation (not necessarily the one used in~$\Lambda^*$) together with an insertion sequence for each level,
which allows us to successfully apply the algorithm corresponding to \cref{cl:generate-drawing-from-ear-orientation}.
We remark that if the algorithm corresponding to \cref{cl:generate-drawing-from-ear-orientation} successfully terminates, it is guaranteed to return an insertion sequence for the given valid ear orientation.
Moreover, we emphasize that \cref{cl:generate-drawing-from-ear-orientation} does not require that the ear orientations are the ones used in~$\Lambda^*$, meaning that our strategy is not invalidated by the fact that algorithm corresponding to \cref{cl:insertion-sequence-exists} might output an insertion sequence even if the given valid ear orientation is not the one used in~$\Lambda^*$.
The total runtime of the algorithm is $2^{\mathcal O(k^2\log k)}n^{\mathcal O(1)}$, as claimed.
\end{claimproof}

\Main*
\label{thm:main-result*}

\begin{proof}
In view of \cref{lem:isolated}, we may assume that~$\mathcal G$ has no isolated vertices.
Assume that there is a constrained level planar drawing~$\Gamma^*$ of~$\mathcal G$. 
We begin by applying \cref{lem:step1}, that is, we construct a family~$\mathcal F$ of $2^{\mathcal O(k\log k)}$ drawings 
for which there exists a refined visibility extension $\Lambda^*$ of $\Gamma^*$ such that the subdrawing~$\Lambda_{\mathrm{core}}$ of $\Lambda^*$ induced by the core of~$\Lambda^*$ with respect to~$C$ is contained in~$\mathcal{F}$.
For each drawing in~$\mathcal F$, we try to apply the algorithm corresponding to \cref{lem:insertTransition} to insert the transition vertices in polynomial time.
If successful, we then try to apply the algorithm corresponding to \cref{lem:insertLeftovers} to construct the desired drawing in $2^{\mathcal O(k^2\log k)}\cdot n^{\mathcal O(1)}$ time.
If successful, we report the drawing.
If we fail at any point along the way, we just continue with the next drawing in~$\mathcal F$.
Given that~$\mathcal F$ contains $\Lambda_\mathrm{core}$, we are guaranteed to eventually obtain the desired drawing.
In contrast, if~$\mathcal G$ does not admit a constrained level planar drawing, the described algorithm is guaranteed to fail in every iteration, in which case we correctly report that the desired drawing does not exists.
The total runtime is $2^{\mathcal O(k^2\log k)}\cdot n^{\mathcal O(1)}$.
\end{proof}

\end{document}